\documentclass[lettersize,journal,onecolumn]{IEEEtran}
\usepackage{amsmath,amsfonts}
\usepackage{algorithmic}
\usepackage{algorithm}
\usepackage{array}
\usepackage[caption=false,font=normalsize,labelfont=sf,textfont=sf]{subfig}
\usepackage{textcomp}
\usepackage{stfloats}
\usepackage{url}
\usepackage{verbatim}
\usepackage{graphicx}
\usepackage{cite}
\usepackage[utf8]{inputenc}
\usepackage{amsmath,amssymb,amsfonts}
\usepackage{mathtools}
\usepackage{amsthm}
\usepackage{algorithmic}
\usepackage{textcomp}
\usepackage{tikz}
\usepackage{cuted}
\usepackage{pgfgantt}
\usepackage{pdflscape}
\usepackage{pst-plot}
\usepackage{comment} 
\usepackage{cases}
\usepackage{nccfoots}
\usepackage{lineno,hyperref}
\usetikzlibrary{spy}
\usetikzlibrary{positioning,calc}
\usetikzlibrary{decorations.pathmorphing,calc,shapes,shapes.geometric,patterns}
\usetikzlibrary{shapes.multipart}
\usepackage{xfrac}
\usepackage{colortbl}
\usepackage{cancel} 
\usetikzlibrary{arrows,positioning,calc,intersections}
\usetikzlibrary{datavisualization.formats.functions}
\def\BibTeX{{\rm B\kern-.05em{\sc i\kern-.025em b}\kern-.08em
    T\kern-.1667em\lower.7ex\hbox{E}\kern-.125emX}}
    
\usepackage{romannum}
\usepackage{pgfplots}
\usepgfplotslibrary{fillbetween}
\usetikzlibrary{arrows, decorations.markings}
\usetikzlibrary{arrows.meta}

\newtheorem{theorem}{Theorem}
\newtheorem*{theorem*}{Theorem}

\newtheorem{definition}[theorem]{Definition}

\newtheorem{corollary}[theorem]{Corollary}

\newcommand{\setx}{\ensuremath{\mathcal{X}}}
\newcommand{\sety}{\ensuremath{\mathcal{Y}}}

\newcommand{\setd}{\ensuremath{\mathcal{D}}}

\newcommand{\sets}{\ensuremath{\mathcal{S}}}

\newcommand{\pmax}{P_{\text{max}}}
\newcommand{\pavg}{P_{\text{avg}}}

\newcommand{\circlearrow}{}
\DeclareRobustCommand{\circlearrow}{%
  \mathrel{\vphantom{\rightarrow}\mathpalette\circle@arrow\relax}%
}
\newcommand{\circle@arrow}[2]{%
  \m@th
  \ooalign{%
    \hidewidth$#1\circ\mkern1mu$\hidewidth\cr
    $#1-$\cr}%
}

\DeclarePairedDelimiterX{\infdivx}[2]{(}{)}{%
  #1\;\delimsize\|\;#2%
}
\newcommand{\infdiv}{D\infdivx}

\DeclareMathOperator*{\pliminf}{p-liminf}
\DeclareMathOperator*{\plimsup}{p-limsup}

\begin{document}
\title{Secure Event-triggered Molecular Communication—Information Theoretic Perspective and Optimal Performance} 



\author{
\IEEEauthorblockN{}}
\author{
\IEEEauthorblockN{Wafa Labidi \IEEEauthorrefmark{1}\IEEEauthorrefmark{5}, Vida Gholamian\IEEEauthorrefmark{2}\IEEEauthorrefmark{5}, Yaning Zhao\IEEEauthorrefmark{2}\IEEEauthorrefmark{5}, Christian Deppe\IEEEauthorrefmark{2}\IEEEauthorrefmark{5} and 
Holger Boche\IEEEauthorrefmark{1}\IEEEauthorrefmark{4}\IEEEauthorrefmark{5}\IEEEauthorrefmark{6}\IEEEauthorrefmark{7} }
\vskip10pt
\IEEEauthorblockA{\IEEEauthorrefmark{1}Technical University of Munich, Chair of Theoretical Information Technology, Munich, Germany\\
\IEEEauthorrefmark{2}Technische Universität Braunschweig, Institute for Communications Technology, Brunswick, Germany\\
\IEEEauthorrefmark{4}Cyber Security in the Age of Large-Scale Adversaries–
Exzellenzcluster, Ruhr-Universit\"at Bochum\\
\IEEEauthorrefmark{5}BMFTR Research Hub 6G-life, Germany\\
\IEEEauthorrefmark{6}{\color{black}{ Munich Center for Quantum Science and Technology (MCQST) }}\\
\IEEEauthorrefmark{7} {\color{black}{Munich Quantum Valley (MQV)}} \\
Email: wafa.labidi@tum.de, v.gholamian-sefiddarboni@tu-bs.de, yaning.zhao@tu-bs.de, christian.deppe@tu-bs.de, boche@tum.de}
}

\maketitle


\begin{abstract}
 Molecular Communication (MC) is an emerging field of research focused on understanding how cells in the human body communicate and exploring potential medical applications. In theoretical analysis, the goal is to investigate cellular communication mechanisms and develop nanomachine-assisted therapies to combat diseases. Since cells transmit information by releasing molecules at varying intensities, this process is commonly modeled using Poisson channels. In our study, we consider a discrete-time Poisson channel (DTPC).
MC is often event-driven, making traditional Shannon communication an unsuitable performance metric. Instead, we adopt the identification framework introduced by Ahlswede and Dueck. In this approach, the receiver is only concerned with detecting whether a specific message of interest has been transmitted. Unlike Shannon transmission codes, the size of identification (ID) codes for a discrete memoryless channel (DMC) increases doubly exponentially with blocklength when using randomized encoding. This remarkable property makes the ID paradigm significantly more efficient than classical Shannon transmission in terms of energy consumption and hardware requirements.
Another critical aspect of MC, influenced by the concept of the Internet of Bio-NanoThings, is security. In-body communication must be protected against potential eavesdroppers. To address this, we first analyze the DTPC for randomized identification (RI) and then extend our study to secure randomized identification (SRI). We derive capacity formulas for both RI and SRI, providing a comprehensive understanding of their performance and security implications.
\end{abstract}

\section{Introduction}

Molecular Communication (MC) is a rapidly growing field of research in which communication occurs through molecular signaling \cite{nakano2012,farsad2016,pierobon2011noise}. Recent studies have begun exploring MC from an information-theoretic perspective. For instance, \cite{JW19} focuses on modeling, \cite{kuscu2019transmitter} examines modulation, and initial information-theoretic analyses can be found in \cite{gohari2016information, rose2019capacity}.

Bienau et al. \cite{bienau2025molecular} provide a broad overview of theoretical modeling and testbed development in molecular communication, highlighting key concepts and recent advancements. Complementing this, the survey in \cite{bhattacharjee2025exhaled} presents a comprehensive review of exhaled-breath modeling and analysis within the molecular communication framework, covering both fundamental theories and practical applications.

Experimental platforms using biocompatible carriers \cite{bartunik2023channel} and magnetic nanoparticles \cite{bartunik2023development,9493754} have demonstrated the feasibility of MC for healthcare and environmental applications. Accurate channel and noise models, especially for nonlinear and Poisson-based systems \cite{farsad2014channel}, are essential to understanding performance limits. In addition, intersymbol interference (ISI) remains a key challenge; mitigation strategies improving reliability and energy efficiency have been proposed in \cite{tepekule2015isi,7331288}.

In MC, a sender such as a cell, releases molecules into the medium at a certain rate (molecules per second) during a specific interval. These molecules propagate via diffusion and/or advection and may degrade due to enzymatic reactions, while the receiver (a nanomachine or another cell) detects information by counting the received molecules. Since molecular production and release are physically limited, constraints on emission rates must be incorporated into the system model. When the release, propagation, and reception of molecules are statistically independent, the received signal follows Poisson statistics for sufficiently large molecule counts. Thus, the discrete-time Poisson channel (DTPC) provides a natural model for such systems, with average and peak power constraints representing the limitations in molecule production and release \cite{jamali2019channel,gohari2016information,unterweger2018experimental}.

Beyond MC, the DTPC is also fundamental in optical communication, where the transmitter is a laser emitting discrete photons and the receiver is a photodetector recording photon arrivals \cite{shamaiCapacityAchieving}. This analogy further emphasizes the generality and analytical importance of the DTPC across communication domains.

Technological advances include microfluidic MIMO platforms \cite{7397863,10059136,9122607}, closed-loop testbeds \cite{brand2024closed,9252869}, and fully chemical synchronization mechanisms \cite{heinlein2024closing}, bringing MC systems closer to biological realism. In \cite{9793850}, a mathematical model and simulator for molecular communication across the blood–tissue barrier are presented, analyzing factors such as blood flow and emission duration to enhance biomedical Internet of Bio-NanoThings (IoBNT) applications. Furthermore, the increasing need for secure and event-triggered communication has become particularly pressing in IoBNT and medical contexts \cite{9538653,10707359}.

In this context, Haselmayr et al. \cite{haselmayr2019integration} explicitly emphasize secure MC as a future requirement for 6G networks, particularly for applications involving Poisson channels, healthcare, and nano-robotics. Communication systems for healthcare and nano-robotics applications are subject to special security and privacy requirements. It is not clear to the MC how post-quantum cryptography solutions can be implemented using synthetic biology due to limited signal processing performance. Therefore, the direct embedding of post-quantum security into direct physical communication is an interesting approach. These needs align with efforts to design secure identification codes and adaptive architectures suited to molecular environments.

The motivation for using identification theory in MC arises from the event-triggered nature of many MC applications, including drug delivery \cite{muller2004challenges}, cancer treatment \cite{jain1999transport}, and health monitoring \cite{nakano2014molecular}. 

Identification theory was first introduced in the context of communication systems by Ahlswede and Dueck in \cite{Idchannels, Idfeedback}, inspired by earlier work on communication complexity \cite{ja1985identification, yao1979some}. They showed that by changing the receiver’s objective, from decoding one of $2^{nC}$ messages to identifying one of $2^{2^{nC}}$ possible messages, communication efficiency can be dramatically improved. Here, the receiver’s goal is to decide whether a particular message was sent, while the sender remains unaware of which message the receiver is checking.

This identification approach extends to scenarios where multiple receivers observe the same transmitted sequence. However, to maintain a manageable identification error, the number of receivers must not grow excessively (i.e., it should remain linear in the block length). Ahlswede and Dueck's methodology relies on local randomization and hash functions, providing an information-theoretic foundation for identification-based communication.

Following their theoretical work, explicit identification codes were developed in \cite{verdu1993explicit, derebeyouglu2020performance, von2023identification}. Even in cases where receiver-side randomization is not feasible, notable performance gains can still be achieved \cite{SPBD21pc, SPBD21f}. For instance, in Gaussian \cite{SPBD21pc} and Poisson channels \cite{SPBD21p, SV23isi}, the number of identifiable messages scales as $n^{nR}$, demonstrating the potential of identification theory beyond traditional Shannon-based communication frameworks.

Furthermore, deterministic identification (DI) in MC has been analyzed in \cite{SV23b,SV23isi,SPBD21p}, where capacity limits for deterministic identification were derived.

Until now, DI \cite{10285119} has been the primary focus of analysis, as it has often been assumed that local randomization is not feasible in MC. However, randomization at the receiver can provide significant performance gains for message identification. If the identification transmitter is located outside the body or bloodstream \cite{khaloopour2022joint}, randomized identification (RI) becomes feasible and worth further exploration.

In this work, we investigate randomized identification (RI) over the DTPC. To the best of our knowledge, the RI capacity of the DTPC has not been studied before. We derive its RI capacity under both average and peak power constraints, which account for the sender’s limited molecule production and release rate.

 Quantitatively, deterministic identification (DI) achieves a super-exponential growth in the number of identifiable messages, $ M_{\mathrm{DI}} \sim 2^{n \log n R} $, whereas randomized identification (RI) achieves a double-exponential growth, $M_{\mathrm{RI}} \sim 2^{2^{n R}} $. This highlights the substantial advantage of randomization in terms of achievable identification rates.


We then extend our analysis to the DTPC with independent and identically distributed (i.i.d.) state sequences, where the random channel state is unknown to both the sender and the receiver. We show that, in this scenario, the RI capacity of the state-dependent DTPC remains unchanged and coincides with its RI capacity in the state-independent case.

Additionally, we explore the security aspects of the discrete-time Poisson wiretap channel (DTPWC) by analyzing its identification capacity. The wiretap channel was originally introduced by Wyner, who established the capacity for the degraded case in \cite{wyner75thewire}; this was later extended to the non-degraded case by Csiszár and Körner in \cite{csiszar2003broadcast}. The secure capacity of the DTPWC \cite{soltani2021degraded,shamai1990capacity,bloch2011physical} is predominantly derived under the assumption of a degraded channel and are based on the notion of \emph{weak secrecy}. In modern information-theoretic security, however, stronger secrecy notions, such as \emph{strong secrecy} and \emph{semantic secrecy}, the latter rooted in cryptographic paradigms, have gained broader acceptance. A detailed discussion of these enhanced secrecy criteria can be found in \cite{bellare2012semantic}.

For our analysis of the secure identification capacity of the degraded DTPWC, we adopt the strong secrecy criterion. As a result, it is essential to establish that the previously known capacity formula remains valid under this stronger secrecy requirement. We derive an explicit expression for the secure identification capacity, thereby shedding light on the fundamental limits of privacy-preserving molecular communication systems.


Although the analysis in this paper is theoretical, the parameters of the adopted Poisson model are consistent with realistic molecular communication environments. For instance, nano-transmitters can release molecules at rates ranging from $10^3$ to $10^6$ molecules per second, while diffusion coefficients for small molecules in aqueous media are typically between $10^{-9}$ and $10^{-10},\text{m}^2/\text{s}$. Moreover, enzymatic degradation and biochemical reactions influence the effective lifetime of signaling molecules. These values correspond to practical biomedical applications such as controlled drug release pathogen detection and biomarker monitoring, showing that the theoretical framework captures biologically plausible conditions.

Finally, the wiretap setup considered here serves as a theoretical model for secure molecular communication. While a literal eavesdropper in an in-vivo environment may not be directly realizable, similar situations can arise in practice, for example, unintended molecule reception, cross-network interference, or signal exchange between wearable and implantable devices. From an information-theoretic perspective, these cases correspond to a wiretap channel where part of the transmitted signal is observable by another entity. Hence, the DTPWC offers a simplified yet meaningful model for studying confidentiality in molecular communication systems.

\textit{Outline:} The rest of the paper is organized as follows. Section~\ref{sec:Preli} introduces the system models, provides key definitions related to ID, and presents the main results. In Section \ref{strong}, Strong secrecy is demonstrated for the DTPWC, with the information leakage to the eavesdropper asymptotically approaching zero as the code length tends to infinity. Section~\ref{sec:Proof} establishes the proof of the RI capacity formula for the DTPC and the state-dependent DTPC under average and peak power constraints. Section~\ref{sec:Secureproof} presents the proof of the SRI capacity formula. Finally, Section~\ref{sec:conclusion} concludes the paper.


\begin{table}[!t]
\caption{Summary of Notations}
\label{tab:notations}
\centering
\resizebox{\columnwidth}{!}{%
\begin{tabular}{ll}
\hline
\textbf{Symbol} & \textbf{Description} \\ \hline
$X$ & Channel input (number of emitted molecules per time slot) \\
$Y$ & Output of the main receiver (Bob) \\
$Z$ & Output of the eavesdropper (Eve) \\
$\alpha_B, \lambda_B$ & Channel gain and background noise for Bob \\
$\alpha_E, \lambda_E$ & Channel gain and background noise for Eve \\
$\Delta$ & Time slot duration \\
$P_{\max}, P_{\text{avg}}$ & Peak and average input power constraints \\
$p_{Y|X}, p_{Z|X}$ & Conditional PMFs of Bob’s and Eve’s Poisson channels \\
$C(W)$ & Transmission capacity of the main Poisson channel \\
$C_S(W,V)$ & Secrecy capacity of the Poisson wiretap channel \\
$C_{ID}(W)$ & Identification capacity of the Poisson channel \\
$C_{SID}(W,V)$ & Secure identification capacity of the Poisson wiretap channel \\
$\mathcal{M}$ & Set of identification messages \\
$n$ & Blocklength of the identification code \\
$P_i$ & Output distribution at Eve corresponding to message $i$ \\
$\lambda_n$ & Error parameter (probability of misidentification) \\
$I(M;Z^n)$ & Mutual information between message and Eve’s observation \\
$\mathrm{TV}(\cdot,\cdot)$ & Total variation distance between two distributions \\
$D(P\|Q)$ & Kullback--Leibler divergence between $P$ and $Q$ \\
$\mathcal{C}_{RI}$ & Randomized identification code \\
$\mathcal{C}_{SRI}$ & Secure randomized identification code \\
$P_X$ & Distribution of a random variable $X$ \\
$\mathcal{P}(\mathcal{X})$ & Set of all probability distributions on $\mathcal{X}$ \\
$|\mathcal{X}|$ & Cardinality of the finite set $\mathcal{X}$ \\
$H(P_X)$ & Shannon entropy of $X$ \\
$\mathbb{E}[X]$, $\mathrm{Var}[X]$ & Expectation and variance of $X$ \\
$I(X;Y)$ & Mutual information between random variables $X$ and $Y$ \\
$\mathcal{X}^c$, $\mathcal{X}-\mathcal{Y}$ & Complement and set difference \\
$\{X_n\}_{n=1}^{\infty}$ & Random process \\
-- & All logarithms and information quantities are to the base 2 \\
\hline
\end{tabular}%
}
\end{table}

\section{System Models and Main Results} \label{sec:Preli}

In this section, we describe the system model, introduce the necessary definitions, review known results, and present our main findings. The first subsection focuses on the DTPC, both with and without a state-dependent channel. In the second subsection, we turn to the DTPWC and present our main results for this channel.

The discrete-time Poisson channel adopted in this work represents an idealized molecular communication model that is analytically tractable and suitable for theoretical investigation. In this model, we neglect molecular-specific noise sources such as background noise fluctuations, counting noise, and intersymbol interference (ISI) resulting from molecule diffusion and residual concentration. These assumptions lead to a memoryless Poisson process, which has been widely used in both optical and molecular communication studies as a first-order approximation. The purpose of this simplification is to focus on the fundamental limits of randomized and secure identification rather than on detailed physical modeling. Future work will address more realistic molecular environments including ISI and stochastic molecular noise.

\subsection{The DTPC}
Let a memoryless DTPC $(\setx, \sety, W(y|x))$ consisting of input alphabet $\setx \subset\mathbb{R}_0^+$, output alphabet $\sety\subset \mathbb{Z}_0^+$ and a pmf $W(y|x)$ on $\sety$, be given.
For $n$ channel uses, the transition probability law is given by
\begin{align*}
    W^n(y^n|x^n)&= \prod_{t=1}^n W(y_t|x_t) \\
    & = \prod_{t=1}^n \exp\big(-(x_t+\lambda_0)\big) \frac{(x_t+\lambda_0)^{y_t}}{{y_t}!}, 
\end{align*}
where $\lambda_0$ is some nonnegative constant, called dark current. The dark current $\lambda_0$, here, represents the non-ideality of the detector. The sequences $x^n=(x_1,x_2,\ldots,x_n) \in \setx^n$ and $y^n=(y_1,y_2,\ldots,y_n) \in \sety^n$ are the channel input and the channel output, respectively. 
 The peak and average power constraints on the input are 
\begin{align}
    & x_t\leq \pmax, \quad \forall t=1,\ldots,n, \label{eq:peakPower}\\
    & \frac{1}{n} \sum_{t=1}^n x_t \leq \pavg, \label{eq:averagePower}
\end{align}
 where $P_{\text{max}}, \ P_{\text{avg}}>0$ represent the values for peak power and average power constraints, respectively. 
 Let $\setx(P_{\text{max}},P_{\text{avg}})$ denote the set of all input symbols $x \in \setx$ such that \eqref{eq:peakPower} and \eqref{eq:averagePower} are satisfied.
 
In the following, the transmission code and the rate of transmission codes are defined.

\begin{definition}
An $(n, N, \lambda)$ transmission code is a family of pairs
\[
\left\{(Q(\cdot|i), D_i),\ i = 1, \ldots, N \right\}
\]
such that for all $i \in \{1, \ldots, N\}$ and some $\lambda \in (0,1)$ we have:

\begin{equation}
    Q(\cdot|i) \in \mathcal{P}(\mathcal{X}^n), \quad D_i \subset \mathcal{Y}^n 
\end{equation}

\begin{equation}
    \mu^{(i)} = \sum_{x^n \in \mathcal{X}^n} Q(x^n|i) W^n(D_i^c | x^n) \leq \lambda 
\end{equation}

\begin{equation}
     D_i \cap D_j = \emptyset, \quad i \ne j 
\end{equation}

\end{definition}

\begin{definition}
\begin{enumerate}
\item The rate \( R \) of an  $(n, N, \lambda)$ transmission code for the channel  W is given by:
\begin{equation}
\label{eq:Rate1}
R = \frac{\log\log(N)}{n} \text.
\end{equation}  

\item The transmission rate \( R \) for the channel W is said to be achievable if, for \( \lambda \in (0, 1) \),  there exists an \((n, N, \lambda)\) transmission code for the channel W.

\item The transmission capacity \(C(W, P_{max}, P_{avg}) \) is the supremum of all achievable rates.
\end{enumerate}
\end{definition}

 Let $C(W, \pmax,\pavg)$ denote the Shannon transmission capacity of the DTPC $W$ under peak and average power constraints $\pmax$ and $\pavg$ as described in \eqref{eq:peakPower} and \eqref{eq:averagePower}, respectively. Then  $C(W, \pmax,\pavg)$ is given by \cite{shamaiCapacityAchieving}
\begin{equation}
    C(W,\pmax,\pavg)  =\max_{\substack{P_X \in \mathcal{P}(\setx) \\ X \in \setx(P_{\text{max}},P_{\text{avg}})}} I(X;Y). \label{eq:ShannonCapacityDTPC}
\end{equation}
For the sake of notational simplicity, we denote the DTPC described above by $W$. \\
In the setting depicted in Fig. \ref{figSystem}, the sender  releases molecules (ID message $i \in \mathcal{N}:=\{1,2,\ldots,N\}$) into the DTPC $W$ at a specific rate over a time interval. The receiver, here, (nanomachines, cells) wants to check whether a specific pathological biomarker exists around the target or not.  Assuming that the release propagation and the reception of individual molecules are statistically similar but independent, the received signal follows Poisson statistics.

In the following, we define RI codes for the DTPC introduced above.
\begin{definition}
	An $(n,N,\lambda_1,\lambda_2)$ RI code with $\lambda_1+\lambda_2<1$ for the DTPC $W$ is a family of pairs 
	$\{(Q(\cdot|i),\setd_i), \quad   i=1,\ldots,N\}$ with 
	\begin{equation*}
	Q(\cdot|i) \in \mathcal{P}(\setx^n(P_{\text{max}},P_{\text{avg}})), \ \setd_i \subset \sety^n,\ \forall i=1,\ldots,N,
	\end{equation*}
	such that the errors of the first kind and the second kind are bounded as follows.
	\begin{align}
	 \mu_1^{(i)} &\triangleq \int_{x^n \in \setx^n(P_{\text{max}},P_{\text{avg}})} Q(x^n|i) \nonumber\\
  & \quad \cdot W^n(\setd_i^c|x^n) dx^n  \leq \lambda_1,  \forall i, \label{eq:firsterror}\\
\mu_2^{(i,j)} & \triangleq \int_{x^n \in \setx^n(P_{\text{max}},P_{\text{avg}})} Q(x^n|i) \nonumber\\
& \quad \cdot W^n(\setd_j|x^n) dx^n  \leq \lambda_2, \forall i\neq j, \label{eq:seconderror} 
	\end{align}
 where $\setx^n(P_{\text{max}},P_{\text{avg}})=\underbrace{\setx(P_{\text{max}},P_{\text{avg}})\times \ldots \times \setx(P_{\text{max}},P_{\text{avg}})}_{n \text{ times}}$.
 
\end{definition}
The error $\mu_1^{(i)}$ in \eqref{eq:firsterror} is called the error of the first kind, which is produced by channel noise and fits the same error definition as for a transmission code. In addition, we have another kind of error $\mu_2^{(i,j)}$ called the error of the second kind \eqref{eq:seconderror}, which results from the ID code construction. In contrast to transmission code, we permit a degree of overlap among the decoding sets in the case of ID codes. This results in both a double exponential increase in the rate (e.g., for DMCs) and introduces another kind of error as described in \eqref{eq:seconderror}. \\

In the following, we define achievable ID rate and ID capacity for our system model.
\begin{definition}
	\begin{enumerate}
		\item The rate $R$ of an $(n,N,\lambda_1,\lambda_2)$ RI code for the channel $W$ is $R=\frac{\log\log(N)}{n}$ bits.
		\item The ID rate $R$ for $W$ is said to be achievable if for $\lambda \in (0,\frac{1}{2})$ there exists an $n_0(\lambda)$, such that for all $n\geq n_0(\lambda)$ there exists an $(n,2^{2^{nR}},\lambda,\lambda)$ RI code for $W$.
		\item The RI capacity $C_{ID}(W,\pmax,\pavg)$ of the channel $W$ under the peak and average power constraints \eqref{eq:peakPower} and \eqref{eq:averagePower} is the supremum of all achievable rates.
	\end{enumerate}
\end{definition}
The following Theorem characterizes the RI capacity of the DTPC $W$ under peak and average power constraints $\pmax$ and $\pavg$, respectively.
\begin{theorem} \label{theorem:randomizedID}
 The RI capacity $C_{ID}(W,\pmax,\pavg)$ of the channel $W$ under peak and average power constraints $\pmax$ and $\pavg$, respectively, is given by
	\begin{equation}\label{eq:theorem}
	C_{ID}(W,\pmax,\pavg)=C(W,\pmax,\pavg).
	\end{equation}
\end{theorem}
It is to be noted that the right-hand side of \eqref{eq:theorem} is a convex optimization problem. Therefore, Theorem \ref{theorem:randomizedID} provides the formula for computing the ID capacity of the DTPC as a function of the peak power and average power values $\pmax$ and $\pavg$, respectively. Many studies have been dedicated to investigating the properties of the capacity-achieving distribution of the DTPC. In the absence of input constraints, the transmission capacity of the DTPC is infinite. The DTPC was addressed in \cite{dytso2021properties} when only subject to peak power constraint and
was shown that the support size is of an order between $\sqrt{\pmax}$ and $\pmax \ln^2 \pmax$. An Analytical formulation of the transmission capacity under only average power constraints is still an open problem. However, bounds and asymptotic behaviors for the DTPC in different scenarios were established e.g., \cite{lapidoth2008capacity}, \cite{lapidoth2008poisson}, \cite{aminian2015capacity}, etc. It was shown in \cite{shamaiCapacityAchieving} that the capacity-achieving distribution for the DTPC under an average-power constraint and/or a peak
power constraint is discrete. However, an analytical formula for the capacity-achieving distribution of the DTPC remains an unresolved question.
All these results can be extended and applied to the RI case.   
\begin{figure}
    \centering
    \resizebox{1\linewidth}{!}{\tikzstyle{farbverlauf} = [ top color=white, bottom color=white]
\tikzstyle{block} = [draw,top color=white, bottom color=white, rectangle, rounded corners,
minimum height=2em, minimum width=2.5em]
\tikzstyle{input} = [coordinate]
\tikzstyle{sum} = [draw, circle,inner sep=0pt, minimum size=2mm,  thick]
\scalebox{1}{
\tikzstyle{arrow}=[draw,->] 
\begin{tikzpicture}[auto, node distance=2cm,>=latex']
\node[] (M) {$i \in \mathcal{N}$};
\node[block,right=.5cm of M] (enc) {Encoder};
\node[block, right=1.7cm of enc] (channel) {Poisson channel $W$};
\node[block,below=.7cm of channel](state){$P_S$};
\node[block, right=1cm of channel] (dec) {Decoder};
\node[right=.5cm of dec] (Mhat) {\begin{tabular}{c} Is ${i}^\prime$ sent? \\ Yes or No? \end{tabular}};
\node[input,right=.5cm of channel] (t1) {};
\node[input,above=1cm of t1] (t2) {};
\draw[-{Latex[length=1.5mm, width=1.5mm]},thick] (M) -- (enc);
\draw[-{Latex[length=1.5mm, width=1.5mm]},thick] (enc) --node[above]{ $X^n$} (channel);
\draw[-{Latex[length=1.5mm, width=1.5mm]},thick] (channel) --node[above]{$Y^n$} (dec);
\draw[-{Latex[length=1.5mm, width=1.5mm]},thick] (dec) -- (Mhat);
\draw[-{Latex[length=1.5mm, width=1.5mm]},thick] (state)--(channel);
\end{tikzpicture}}} 
    \caption{Discrete-time memoryless Poisson channel with random state}
    \label{figSystem}
\end{figure}
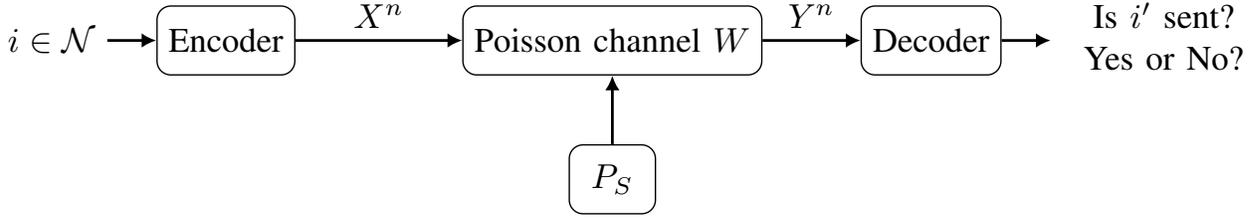

Let a DTPC with random state $(\setx\times \sets, W_S(y|x,s), \sety)$ consisting of an input alphabet $\setx \subset \mathbb{R}_0^+$, a discrete output alphabet $\sety \subset \mathbb{Z}_0^+$, a finite state set $\sets$ and a pmf $W(y|x,s)$ on $\sety$, be given. 
The channel is memoryless, i.e., the probability for a sequence $y^n \in \sety^n$ to be received if the input sequence $x^n \in \setx^n$ was sent and the sequence state is $s^n \in \sets^n$ is given by 
		\begin{equation*}
		W_S^n(y^n|x^n,s^n)=\prod_{t=1}^n W_S(y_t|x_t,s_t).
		\end{equation*}
		 The state sequence $(S_1,S_2,\ldots,S_n)$ is i.i.d. according to the distribution $P_S$. We assume that the input $X_t$ and the state $S_t$ are statistically independent for all $t\in\{1,2,\ldots,n\}$. We further assume that the input satisfies the peak power and average power constraints in \eqref{eq:peakPower} and \eqref{eq:averagePower}, respectively. We further assume that the channel state is  known to neither the sender nor the receiver. Now, let's define ID codes for the state-dependent DTPC $W_S$. 
\begin{definition}
	An $(n,N,\lambda_1,\lambda_2)$ RI code with $\lambda_1+\lambda_2<1$ for the channel $W_S$ is a family of pairs 
	$\{(Q(\cdot|i),\setd_i(s^n))_{s^n \in \sets^n}, \quad   i=1,\ldots,N\}$ with 
	\begin{equation*}
	Q(\cdot|i) \in \mathcal{P}(\setx^n(P_{\text{max}},P_{\text{avg}})), \ \setd_i(s^n) \subset \sety^n \end{equation*}
 for all $s^n \in \sets^n$ and for all $i=1,\ldots,N$
	such that the errors of the first kind and the second kind are bounded as follows.
	\begin{align*}
	&\sum_{s^n \in \sets^n} P_S^n(s^n) \int_{x^n \in \setx^n(P_{\text{max}},P_{\text{avg}})} Q(x^n|i) \\
 & \quad \cdot W_S^n(\setd_i(s^n)^c|x^n,s^n) dx^n  \leq \lambda_1,  \forall i,\\
	&\sum_{s^n \in \sets^n} P_S^n(s^n) \int_{x^n \in \setx^n(P_{\text{max}},P_{\text{avg}})} Q(x^n|i) \\
 & \quad \cdot W_S^n(\setd_j(s^n)|x^n,s^n) dx^n \leq \lambda_2, \forall i\neq j. 
	\end{align*}

\end{definition}
In the following, we define achievable ID rate and ID capacity for the state-dependent DTPC $W_S$.
\begin{definition}
	\begin{enumerate}
		\item The rate $R$ of an $(n,N,\lambda_1,\lambda_2)$ RI code for the channel $W_S$ is $R=\frac{\log\log(N)}{n}$ bits.
		\item The ID rate $R$ for $W_S$ is said to be achievable if for $\lambda \in (0,\frac{1}{2})$ there exists an $n_0(\lambda)$, such that for all $n\geq n_0(\lambda)$ there exists an $(n,2^{2^{nR}},\lambda,\lambda)$ randomized ID code for $W_S$.
		\item The RI capacity $C_{ID}(W_S,\pmax,\pavg)$ of the channel $W_S$ under peak and average power constraints \eqref{eq:peakPower} and \eqref{eq:averagePower} is the supremum of all achievable rates.
	\end{enumerate}
\end{definition}
The following corollary, a consequence of Theorem \ref{theorem:randomizedID}, outlines the RI capacity of the state-dependent DTPC $W_S$.
\begin{corollary} \label{corollary}
The RI capacity of the state-dependent DTPC $W_S$ under peak and average power constraints $\pmax$ and $\pavg$, respectively, is given by 
\begin{equation*}
    C_{ID}(W_S,\pmax,\pavg) =\max_{\substack{P_X \in \mathcal{P}(\setx) \\ X \in \setx(P_{\text{max}},P_{\text{avg}})}} I(X;Y).
\end{equation*}
\end{corollary}
\begin{proof}
Let $W^a(y|x) = \sum_{s \in \mathcal{S}} P_S(s) W_S(y|x,s)$ be the channel obtained by averaging the state-dependent DTPCs $W_S$ over the finite set of states $\mathcal{S}$.
Then, it can be readily verified that the Shannon capacity of the state-dependent DTPC $W_S$ when the state is available at neither the sender nor the receiver is given by
\begin{equation*}
    C(W_S, \pmax,\pavg) =\max_{\substack{P_X \in \mathcal{P}(\setx) \\ X \in \setx(P_{\text{max}},P_{\text{avg}})}} I(X;Y).
\end{equation*}
As the set of states $\sets$ is finite, we can show that the state-dependent DTPC $W_S$ satisfies the strong converse property as well. That means that the RI capacity of the state-dependent DTPC $W_S$ coincides with its Shannon transmission capacity. This completes the proof of Corollary \ref{corollary}.
\end{proof}
Now, we revisit the definitions of inf-mutual information rate and sup-mutual information rate defined in \cite{HanBook}. Those quantities will play a key role in proving Theorem \ref{theorem:randomizedID}.
We first recall the definition of $\liminf$ and $\limsup$ in probability.
\begin{definition}
For a given sequence $\{X^n\}_{n=1}^\infty$ of RVs, the $\liminf$ and $\limsup$ in probability are defined as
\begin{align*}
    \pliminf_{n\to \infty} X^n & := \sup \big\{ \beta\colon \lim_{n \to \infty} \Pr\{X^n< \beta\}=0 \big\} \\
    \plimsup_{n\to \infty} X^n & := \inf \big\{ \alpha\colon \lim_{n \to \infty} \Pr\{X^n>\alpha\}=0 \big\},
\end{align*}
respectively. We refer the reader to \cite{HanBook} for more details.
\end{definition}
The quantities mentioned above play a pivotal role in establishing information-spectrum methods.
\begin{definition}
    Let $\mathbf{X}=\{X^n\}_{n=1}^\infty$ be a general input process, where $X^n$ is an arbitrary RV with probability distribution $P_{X^n}$ on $\setx^n$. Let $\mathbf\{Y^n\}_{n=1}^\infty$ be the output of a channel $\mathbf{W}=\{W^n\}_{n=1}^\infty$ corresponding to the input $\mathbf{X}$. The output $Y^n$ is an arbitrary RV with probability distribution $P_{Y^n}$ on $\sety^n$ given by
    \begin{align*}
        P_{X^nY^n}(x^n,y^n) &= P_{X^n}(x^n) W^n(y^n|x^n).
    \end{align*}
    The spectral inf-mutual information rate and the spectral sup-mutual information rate defined in \cite{HanBook} are given by \eqref{eq:infMI} and \eqref{eq: supMI}, respectively.
   \begin{align}
  \underline{I}(\mathbf{X},\mathbf{Y}) &= \pliminf_{n\to \infty} \frac{1}{n} \frac{W^n(Y^n|X^n)}{P_{Y^n}(Y^n)}, \label{eq:infMI} \\
  \overline{I}(\mathbf{X},\mathbf{Y}) &= \plimsup_{n\to \infty} \frac{1}{n} \frac{W^n(Y^n|X^n)}{P_{Y^n}(Y^n)}, \label{eq: supMI}
   \end{align}
  where $\frac{1}{n} \frac{W^n(Y^n|X^n)}{P_{Y^n}(Y^n)}$ is called the mutual information density rate of $(\mathbf{X},\mathbf{Y})$ as defined in \cite{HanBook}.
\end{definition}

\subsection{The DTPWC}

This section presents an analysis of the identification capacity of the Discrete-Time Poisson Wiretap Channel (DTPWC). The DTPWC, illustrated in Fig.\ref{fig:DTPWC}, consists of a transmitter (Alice), a legitimate receiver (Bob), and a wiretapper (Eve). Figure \ref{fig:cell} shows the Molecular Wiretap Channel with Transmitter (TX), Receiver (RX), and Eavesdropper (EX) implemented at the cell level.

\begin{figure}
    \centering
    \resizebox{1\linewidth}{!}{\tikzset{every picture/.style={line width=0.75pt}} 

\begin{tikzpicture}[x=0.75pt,y=0.75pt,yscale=-1,xscale=1]

\draw   (221,81) .. controls (221,77.69) and (223.69,75) .. (227,75) -- (354.5,75) .. controls (357.81,75) and (360.5,77.69) .. (360.5,81) -- (360.5,99) .. controls (360.5,102.31) and (357.81,105) .. (354.5,105) -- (227,105) .. controls (223.69,105) and (221,102.31) .. (221,99) -- cycle ;
\draw   (112,79.6) .. controls (112,75.95) and (114.95,73) .. (118.6,73) -- (159.9,73) .. controls (163.55,73) and (166.5,75.95) .. (166.5,79.6) -- (166.5,99.4) .. controls (166.5,103.05) and (163.55,106) .. (159.9,106) -- (118.6,106) .. controls (114.95,106) and (112,103.05) .. (112,99.4) -- cycle ;
\draw   (414.5,80.6) .. controls (414.5,76.95) and (417.45,74) .. (421.1,74) -- (462.4,74) .. controls (466.05,74) and (469,76.95) .. (469,80.6) -- (469,100.4) .. controls (469,104.05) and (466.05,107) .. (462.4,107) -- (421.1,107) .. controls (417.45,107) and (414.5,104.05) .. (414.5,100.4) -- cycle ;
\draw   (409,146.5) .. controls (409,143.19) and (411.69,140.5) .. (415,140.5) -- (542.5,140.5) .. controls (545.81,140.5) and (548.5,143.19) .. (548.5,146.5) -- (548.5,164.5) .. controls (548.5,167.81) and (545.81,170.5) .. (542.5,170.5) -- (415,170.5) .. controls (411.69,170.5) and (409,167.81) .. (409,164.5) -- cycle ;
\draw    (57.5,90.09) -- (109,90.09) ;
\draw [shift={(111,90.09)}, rotate = 180] [color={rgb, 255:red, 0; green, 0; blue, 0 }  ][line width=0.75]    (10.93,-3.29) .. controls (6.95,-1.4) and (3.31,-0.3) .. (0,0) .. controls (3.31,0.3) and (6.95,1.4) .. (10.93,3.29)   ;
\draw    (166.5,89.59) -- (218,89.59) ;
\draw [shift={(220,89.59)}, rotate = 180] [color={rgb, 255:red, 0; green, 0; blue, 0 }  ][line width=0.75]    (10.93,-3.29) .. controls (6.95,-1.4) and (3.31,-0.3) .. (0,0) .. controls (3.31,0.3) and (6.95,1.4) .. (10.93,3.29)   ;
\draw    (361,90.59) -- (412.5,90.59) ;
\draw [shift={(414.5,90.59)}, rotate = 180] [color={rgb, 255:red, 0; green, 0; blue, 0 }  ][line width=0.75]    (10.93,-3.29) .. controls (6.95,-1.4) and (3.31,-0.3) .. (0,0) .. controls (3.31,0.3) and (6.95,1.4) .. (10.93,3.29)   ;
\draw    (469.5,90.59) -- (582,90.59) ;
\draw [shift={(584,90.59)}, rotate = 180] [color={rgb, 255:red, 0; green, 0; blue, 0 }  ][line width=0.75]    (10.93,-3.29) .. controls (6.95,-1.4) and (3.31,-0.3) .. (0,0) .. controls (3.31,0.3) and (6.95,1.4) .. (10.93,3.29)   ;
\draw    (547.5,156.59) -- (585.5,156.59) ;
\draw [shift={(587.5,156.59)}, rotate = 180] [color={rgb, 255:red, 0; green, 0; blue, 0 }  ][line width=0.75]    (10.93,-3.29) .. controls (6.95,-1.4) and (3.31,-0.3) .. (0,0) .. controls (3.31,0.3) and (6.95,1.4) .. (10.93,3.29)   ;
\draw    (381.25,156.09) -- (406,156.09) ;
\draw [shift={(408,156.09)}, rotate = 180] [color={rgb, 255:red, 0; green, 0; blue, 0 }  ][line width=0.75]    (10.93,-3.29) .. controls (6.95,-1.4) and (3.31,-0.3) .. (0,0) .. controls (3.31,0.3) and (6.95,1.4) .. (10.93,3.29)   ;
\draw    (381.25,90.34) -- (381.25,156.09) ;
\draw    (439,51.84) -- (439,71.84) ;
\draw [shift={(439,73.84)}, rotate = 270] [color={rgb, 255:red, 0; green, 0; blue, 0 }  ][line width=0.75]    (10.93,-3.29) .. controls (6.95,-1.4) and (3.31,-0.3) .. (0,0) .. controls (3.31,0.3) and (6.95,1.4) .. (10.93,3.29)   ;

\draw (114,83) node [anchor=north west][inner sep=0.75pt]   [align=left] {{\small Encoder}};
\draw (239,83) node [anchor=north west][inner sep=0.75pt]   [align=left] {{\small Main Channel W}};
\draw (417,83) node [anchor=north west][inner sep=0.75pt]   [align=left] {{\small Decoder}};
\draw (425,149) node [anchor=north west][inner sep=0.75pt]   [align=left] {{\small Wiretap Channel V}};
\draw (17,83) node [anchor=north west][inner sep=0.75pt]    {$i\in N$};
\draw (179,66) node [anchor=north west][inner sep=0.75pt]    {$X^{n}$};
\draw (374,66) node [anchor=north west][inner sep=0.75pt]    {$Y^{n}$};
\draw (434,28) node [anchor=north west][inner sep=0.75pt]    {$j$};
\draw (596,149) node [anchor=north west][inner sep=0.75pt]    {$Z^{n}$};
\draw (589,83) node [anchor=north west][inner sep=0.75pt]   [align=left] {Yes/No};

\end{tikzpicture}} 
    \caption{Degraded DTPWC.}
    \label{fig:DTPWC}
\end{figure}
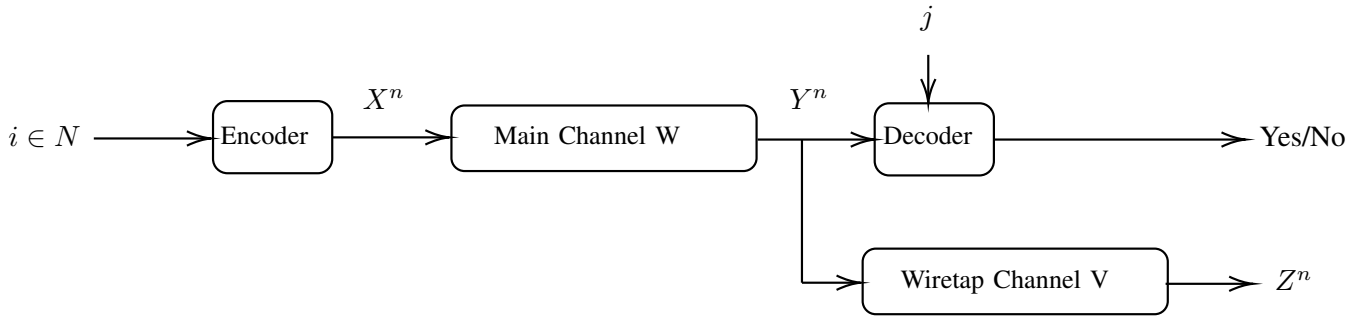

Let \( X \) represent the input alphabet, while \( Y \) and \( Z \) denote the output alphabets for the legitimate receiver and the wiretapper, respectively. The channel transition matrix between Alice and Bob is given by \( W(y|x) \), while \( V(z|x) \) represents the channel transition matrix between Alice and Eve. The wiretap channel is denoted by the pair \( (W, V) \). If both \( W \) and \( V \) are discrete, then the wiretap channel is also discrete. Communication over the wiretap channel must ensure both reliability and secrecy.

The parameters of the considered Poisson model can also be interpreted in the context of realistic molecular communication environments. For instance, nano-transmitters may release information molecules at rates ranging from $(10^3)$ to $(10^6)$ molecules per second, and typical diffusion coefficients of small molecules in aqueous media are on the order of $(10^{-9})–(10^{-10})$ m²/s. Moreover, enzymatic degradation and other biochemical reactions occurring in biological fluids influence the effective lifetime of the molecules. These ranges are consistent with practical biomedical scenarios such as controlled drug delivery and biomarker or pathogen detection, indicating that the proposed model, while theoretical, captures biologically plausible dynamics. The subsequent analysis is therefore based on this simplified but representative framework.

\begin{figure}
    \centering
    \resizebox{1\linewidth}{!}{\input{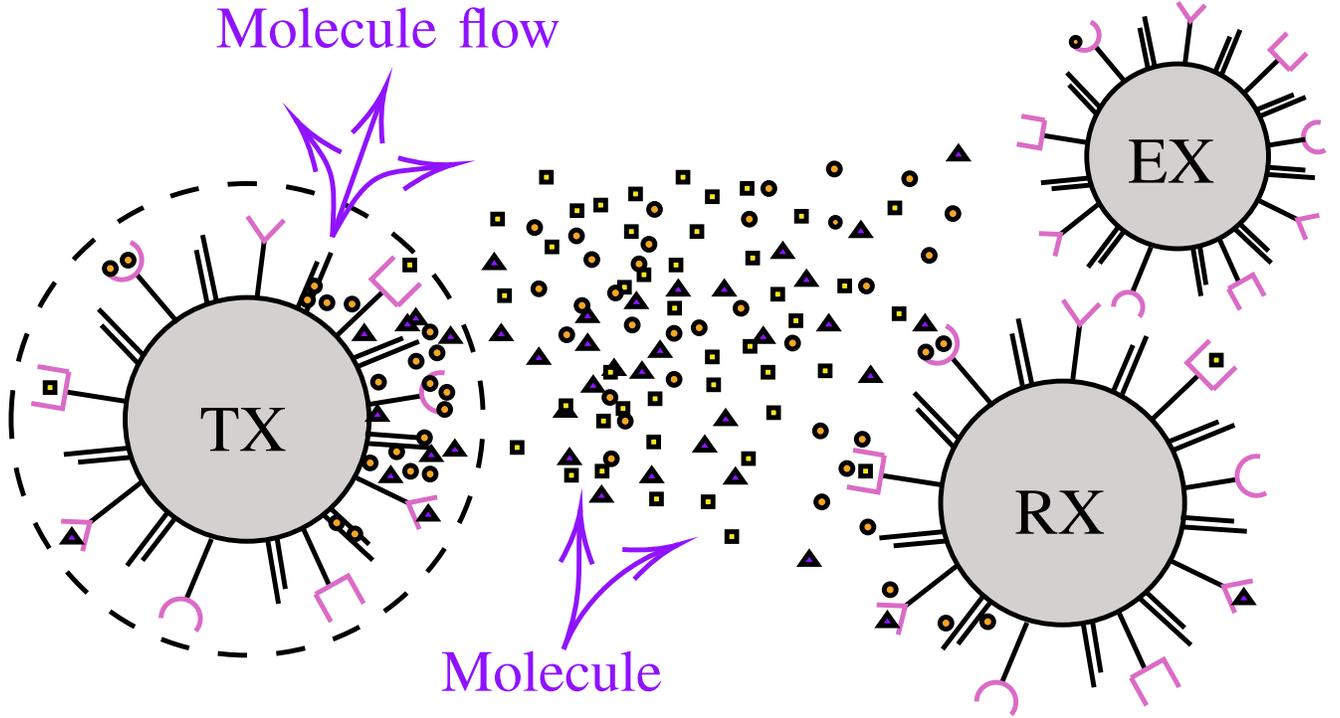}} 
    \caption{Cell-Based Wiretap Model in Molecular Communications}
    \label{fig:cell}
\end{figure}

For \( n \) channel uses, the transition probabilities for the wiretap channel are:


\[(W, V):
\left\{
\begin{aligned}
W^n(y^n|x^n) &= \prod_{t=1}^{n} \exp\bigl(-(x_t + \lambda_B)\bigr) \frac{(x_t + \lambda_B)^{y_t}}{y_t!}, \\
V^n(z^n|x^n) &= \prod_{t=1}^{n} \exp\bigl(-(x_t + \lambda_E)\bigr) \frac{(x_t + \lambda_E)^{z_t}}{z_t!}.
\end{aligned}
\right.
\]

where \( \lambda_B \) and \( \lambda_E \) are the dark currents at the legitimate user's and the eavesdropper’s receivers, respectively. Additionally, \( x^n = (x_1, x_2, \dots, x_n) \) represents the channel input sequence, while \( y^n = (y_1, y_2, \dots, y_n) \) and \( z^n = (z_1, z_2, \dots, z_n) \) denote the observations of Bob and Eve, respectively.

The message set is defined as:

\[
\mathcal{N}= \{1, 2, \dots, N\}.
\]

The channel input is assumed to satisfy the following peak and average power constraints:  
\begin{equation}
\label{eq:powermax} 
x_t \leq P_{\text{max}}, \quad \forall t = 1, \dots, n,
\end{equation}

\begin{equation}
\label{eq:poweravg} 
\frac{1}{n} \sum_{t=1}^{n} x_t \leq P_{\text{avg}}.
\end{equation}
We say that a transmission code for the DTPC is a transmission code for the DTPWC under the \emph{weak secrecy} criterion if it additionally satisfies
\begin{equation}
    \lim_{n\to\infty}\frac{1}{n} I(M;Z) = 0.
\end{equation}

It was shown in \cite{soltani2021degraded} that the capacity of the degraded wiretap channel is given by
\begin{equation}
    \label{eq:Cs}
    C_S(W) = \sup_{X} \left[ I(X;Y) - I(X;Z) \right].
\end{equation}

However, for our proof later, we require a stronger notion of secrecy, known as \emph{strong secrecy}, which is defined by the same condition:
\begin{equation}
    \lim_{n\to\infty} I(M;Z) = 0.
\end{equation}
Nonetheless, as we show in Section~\ref{strong}, the result under weak secrecy implies the result under strong secrecy.

We now proceed with the definition of secrecy for RI codes.

\begin{definition}
A randomized \((n, N, \lambda_1, \lambda_2)\) identification code for the DTPWC \((V, W)\) is defined as a family of pairs \( \{ (Q(\cdot | i), \mathcal{D}_i) , i = 1, \dots, N \} \) such that, for every \( i \in \{1, \dots, N\} \), with some \( \lambda_1 + \lambda_2 < 1 \), and for all \( \varepsilon \subset \mathbb{Z}^n \), the following holds:

\begin{equation}
Q(\cdot | i) \in \mathcal{P}(\mathcal{X}^n(P_\text{max},P_\text{avg} )), \quad \mathcal{D}_i \subset
 \mathcal{Y}^n,\quad \forall i = 1, \dots, N,
\end{equation}

\begin{equation}
\sum_{i=1}^{n} x_i \leq n \cdot P_{avg}, \quad \forall x^n \in \mathcal{X}^n, \quad x^n = (x_1, x_2, \dots, x_n),
\end{equation}

\begin{equation}
\label{eq:Type1}
\int_{x^n \in X^n} Q(x^n | i) W^n(\mathcal{D}_i^c | x^n) \, d^n x^n \leq \lambda_1,
\end{equation}

\begin{equation}
\label{eq:Type2}
\int_{x^n \in X^n} Q(x^n | j) W^n(\mathcal{D}_i | x^n) \, d^n x^n \leq \lambda_2, \quad \forall i \neq j,
\end{equation}

\begin{equation}
\label{eq:Eve}
\begin{split}
&\int_{x^n \in X^n} Q(x^n | j) V^n(\varepsilon | x^n) \, d^n x^n \\
&\quad + \int_{x^n \in X^n} Q(x^n | i) V^n(\varepsilon | x^n) \, d^n x^n \\
&\geq 1 - \lambda, \quad \forall i \neq j,\; \lambda \in (0, 1).
\end{split}
\end{equation}

where:
\begin{equation}
\begin{split}
\label{eq:X(P)}
&\mathcal{X}^n(P_\text{max},P_\text{avg} ) =\\
&\underbrace{\mathcal{X}(P_\text{max},P_\text{avg} ) \times \mathcal{X}(P_\text{max},P_\text{avg} ) \times \cdots \times \mathcal{X}(P_\text{max},P_\text{avg} )}_{n \, \text{times}}.
\end{split}
\end{equation}

Equation \ref{eq:Eve} illustrates that the wiretapper is unable to identify the transmitted message \( i \). It represents the sum of the type I and type II errors for this hypothesis test. Specifically, it implies that either the type I error exceeds \( 1 - \frac{\lambda}{2} \), or there exists some \( j \neq i \) where the type II error exceeds \( 1 - \frac{\lambda}{2} \), or both errors occur simultaneously. This ensures that the wiretapper cannot correctly identify message \( i \). Furthermore, the condition in \ref{eq:Eve} guarantees semantic secrecy, which, in turn, provides strong secrecy assurances for the system:

\begin{equation}
\label{eq:strongsecrecy}
\lim_{n \to \infty} I(M; Z^n) = 0.
\end{equation}

\end{definition}
In the following, the achievable rate and capacity are defined
for the system model.
\begin{definition}
\begin{enumerate}
\item The rate \( R \) of an \((n, N, \lambda_1, \lambda_2)\) Randomized Identification (RI) code for the channel (\( W, V \)) is given by:

\begin{equation}
\label{eq:Rate2}
R = \frac{\log\log(N)}{n} \text.
\end{equation}  

\item The ID rate \( R \) for the channel pair \( (W, V) \) is said to be achievable if, for \( \lambda \in (0, \frac{1}{2}) \), there exists an \( n_0(\lambda) \) such that for all \( n \geq n_0(\lambda) \), there exists an \((n, N, \lambda, \lambda)\) Randomized Identification (RI) code for the channel \( (W, V) \).

\item The RI capacity \( C_{\text{SID}}(W, V) \) is the supremum of all achievable rates.
\end{enumerate}
\end{definition}

\begin{theorem}
\label{theorem:DichotomyTheorem}
Let \( C(W) \) be the Shannon capacity of the channel \( W \) and let \( C_S(W, V) \) be the secrecy capacity of the wiretap channel \( (W, V) \). Then:
\begin{equation}
\label{eq26}
C_{\text{SID}}(W, V) = 
\begin{cases} 
C(W) & \text{if } C_S(W, V) > 0 \\[6pt] 
0 & \text{if } C_S(W, V) = 0
\end{cases}
\end{equation}

Where \( C_{\text{SID}}(W, V) \) is the secure identification capacity of the channel \( (W, V) \) defined as follows:

\begin{equation}
\begin{split}
C_{\text{SID}}(W, V) = \max\big\{ R : \forall \lambda > 0,\;
\exists n(\lambda) \text{ such that for } n \ge n(\lambda), \\
N_S(n, \lambda) \ge 2^{2^{nR}} \big\}
\end{split}
\end{equation}

where \( N_S(n, \lambda) \) is the maximal cardinality such that an \((n, N, \lambda_1, \lambda_2)\) identification wiretap code for the channel \( (W, V) \) exists.

\end{theorem}

In this section, Theorem \ref{theorem:DichotomyTheorem} is extended to the DTPC case.

\begin{theorem}
Let \( C_{\text{SID}}(W, V, P_{\max}, P_{\text{avg}}) \) be the secure identification capacity of the wiretap channel \( (W, V) \) and then:

\begin{equation}
\begin{aligned}
C_{\text{SID}}&(W, V, P_{\max}, P_{\text{avg}}) = \\
&\begin{cases}
C(W, P_{\max}, P_{\text{avg}}) & \text{if } C_{\text{S}}(W, V, P_{\max}, P_{\text{avg}}) > 0, \\
0 & \text{if } C_{\text{S}}(W, V, P_{\max}, P_{\text{avg}}) = 0.
\end{cases}
\end{aligned}
\end{equation}



\end{theorem}

{\section{Strong Secrecy}\label{strong}
This section establishes strong secrecy for the DTPWC. The analysis focuses on the degraded case of the DTPWC, and the resulting capacity corresponds to the strong secrecy capacity under this assumption.

The per-letter peak constraint in \ref{eq34} does not mean that the transmitter uses vanishing power as the blocklength $n$ increases. It is only a technical condition to keep the input intensity bounded for each channel use, while the average power still follows the constraint in \ref{eq:powermax}. This assumption simplifies the strong secrecy proof and ensures that information leakage remains negligible. The secrecy strategy is therefore based on codeword randomization, not on reducing transmission power.

As mentioned, the secrecy capacity $C_s$ is defined Equation \ref{eq:Cs}.
Operate at transmission rates

$$
  R < C_s,
$$

to guarantee reliable and secure communication.


By applying the chain rule of mutual information, the expression is decomposed as follows:

\begin{equation}
  I(M;Z^n) = \sum_{i=1}^n I(M; Z_i \mid Z^{i-1}).
\end{equation}

For memoryless channels, where $Z_i$ depends only on $X_i(M)$, this simplifies to:

\begin{equation}
  I(M; Z^n)= \sum_{i=1}^n I(M; Z_i).
\end{equation}


Since

\begin{equation}
  Z_i \sim \text{Poisson}(X_i + \lambda_E),
\end{equation}

by data processing and KL-divergence properties,

\begin{equation}
  I(M; Z_i) \leq \mathbb{E} \bigl[ D( P_{Z_i|X_i} \| P_{Z_i} ) \bigr].
\end{equation}

Using the KL-divergence between two Poisson distributions, one can bound

\begin{equation}
  D\bigl(P_{Z_i|X_i} \| P_{Z_i}\bigr) \leq \frac{\mathbb{E}[X_i^2]}{2(\lambda_E + \mathbb{E}[X_i])}.
\end{equation}

Under a peak constraint

\begin{equation}
\label{eq34}
  X_i \leq \epsilon \lambda_E,
\end{equation}

this decays roughly as

\begin{equation}
  O(\epsilon^2 \lambda_E).
\end{equation}


Summing over all $n$ symbols gives

\begin{equation}
  I(M; Z^n) \leq n \cdot O(\epsilon^2 \lambda_E).
\end{equation}

By choosing

\begin{equation}
  \epsilon = O\left(\frac{1}{n}\right),
\end{equation}

one obtains

\begin{equation}
  I(M; Z^n) \leq O\left(\frac{\lambda_E}{n}\right) \to 0 \quad \text{as} \quad n \to \infty,
\end{equation}

meaning asymptotic secrecy is achieved.



Show that the total variation (TV) distance between these two distributions satisfies

\begin{equation}
  \text{TV}(P_{Z^n}, P_{Z^n}^0) \to 0 \implies I(M; Z^n) \to 0,
\end{equation}

confirming the output at Eve is indistinguishable from pure noise.


Pinsker’s inequality relates KL-divergence and total variation distance:

\begin{equation}
  \text{TV}(P_{M,Z^n}, P_M P_{Z^n}) \leq \frac{1}{2} \sqrt{D(P_{M,Z^n} \| P_M P_{Z^n})}.
\end{equation}

Since

\begin{equation}
  D(P_{M,Z^n} \| P_M P_{Z^n}) = I(M; Z^n) \to 0.
\end{equation}

The total variation distance also vanishes, proving asymptotic independence between message $M$ and Eve’s observation $Z^n$.

As a consequence of Theorem \ref{eq26}, the following corollary
characterizes the identification capacity of the state-dependent
DTPWC.

\begin{corollary}[Strong Secrecy Identification Capacity for DTPWC]
For theDTPWC, the identification capacity under the strong secrecy criterion is equal to the transmission capacity. Specifically, it holds that:
\[
C_{\text{SID}}^{\text{strong}}(W, P_{\text{max}}, P_{\text{avg}}) = C(W, P_{\text{max}}, P_{\text{avg}}),
\]
where the channel capacity $C(W, P_{\text{max}}, P_{\text{avg}})$ is given by
\[
C(W, P_{\text{max}}, P_{\text{avg}}) = \max_{\substack{P_X \in \mathcal{P}(\mathcal{X}) \\ X \in \mathcal{X}(P_{\text{max}}, P_{\text{avg}})}} I(X; Y).
\]
\end{corollary}

}

\section{Proof of Theorem \ref{theorem:randomizedID}} \label{sec:Proof}
In this section, we show that the RI capacity of the DTPC $W$ coincides with its transmission capacity. Indeed, it has been shown in \cite[Corollary 6.6.1]{HanBook} that if a channel $W$ satisfies the strong converse property under some cost constraint, then the corresponding RI capacity and transmission capacity coincide.   
In the following, we prove that the DTPC under peak power constraint $\pmax$ and average power constraint $\pavg$ satisfies the strong converse property. Let $\mathcal{U}(\pmax,\pavg)$ denote the set of all input processes $\mathbf{X}=\{X^n\}_{n=1}^\infty$ satisfying 
\begin{equation*}
    \Pr\{ X^n \in \setx^n(\pmax,\pavg)\}=1
\end{equation*}
for all $n=1,2, \ldots$. Then, as stated in \cite[Theorem 3.7.1]{HanBook}, the DTPC satisfies the strong converse property under cost constraints $\pmax$ and $\pavg$ if and only if 
\begin{align}
    \sup_{ \mathbf{X} \in \mathcal{U}(\pmax,\pavg)} \underline{I}(\mathbf{X},\mathbf{Y}) =  \sup_{ \mathbf{X} \in \mathcal{U}(\pmax,\pavg)}  \overline{I}(\mathbf{X},\mathbf{Y}).
    \label{eq:Strong}
\end{align}
By applying \cite[Theorem 3.6.1]{HanBook}, the transmission capacity of the DTPC $W$ under peak and average power constraints can be rewritten as follows:
\begin{equation*}
    C(W,\pmax,\pavg) = \sup_{ \mathbf{X} \in \mathcal{U}(\pmax,\pavg)} \underline{I}(\mathbf{X},\mathbf{Y}).
\end{equation*}
Thus, in order to develop the strong converse property for the DTPC $W$ under cost constraints $\pmax$ and $\pavg$, in view of \cite[Theorem 3.6.1]{HanBook} it suffices to show 
\begin{align}
    \sup_{ \mathbf{X} \in \mathcal{U}(\pmax,\pavg)}  \overline{I}(\mathbf{X},\mathbf{Y}) & \leq \sup_{\substack{P_X \in \mathcal{P}(\setx) \nonumber \\ X \in \setx(P_{\text{max}},P_{\text{avg}}) }} I(X;Y) \\
    & =\max_{\substack{P_X \in \mathcal{P}(\setx) \\ X \in \setx(P_{\text{max}},P_{\text{avg}})}} I(X;Y) \nonumber \\
    & = C(W,\pmax,\pavg).\label{eq:StrongConverseProperty}
\end{align}

Let $P_X \in \mathcal{P}\left(\setx(\pmax,\pavg)\right)$ an arbitray input distribution and $P_Y \in \mathcal{P}(\sety)$ the corresponding output distribution via the DTPC $W$. Then, for a fixed input sequence $x^n=(x_1,x_2,\ldots,x_n) \in \setx^n(\pmax,\pavg)$, we define $I({Y},x)$ as follows
\begin{align*}
    I({Y},x) & =\log \frac{W(Y|x)}{P_{{Y}}(Y)}.
\end{align*}
The conditional expectation of $I({Y},x)$ given $X=x$ is given by
\begin{align}
    \mathbb{E}[I({Y},x)]&= \sum_{y=0}^\infty W(y|x) \log \frac{W(y|x)}{P_{{Y}}(y)} \nonumber \\
    &= \infdiv{W(\cdot|x)}{P_{{Y}}}. \label{eq:expectationYi}
\end{align}
The mutual information induced by the input distribution $P_{X}$ is given by
\begin{align*}
    I(X,Y)&= \int_{0}^{\pmax} P_X(x)  \mathbb{E}[I({Y},x)] dx \\
    &=  \int_{0}^{\pmax} P_X(x)  \infdiv{W(\cdot|x)}{P_{{Y}}} dx.
\end{align*}

In the following, we define the Lagrangian function $L(\mu,x,P_{{X}})$ with Lagrange multiplier $\mu \geq 0$ as
\begin{align}
    L(\mu,x,P_{{X}}) &= I({X},{Y}) + \mu(x-\pavg)\nonumber\\ 
    & \quad -\infdiv{W(\cdot|x)}{P_{{Y}}}.
\end{align}
It has been shown in \cite{shamaiCapacityAchieving} and \cite[Theorem 5]{capacityAchievingDistribution} that Kuhn-Tucker conditions yield the following theorem.
\begin{theorem}[\cite{capacityAchievingDistribution}]
The distribution $P_{{X}} \in \mathcal{P}\left( \setx(\pmax,\pavg)\right)$ is capacity-achieving iff for some $\mu \geq 0$, the following conditions are satisfied 
\begin{align}
      L(\mu,x,P_{{X}})& \geq 0, \quad x \in [0,\pmax], \\
        L(\mu,x,P_{{X}}) & = 0, \quad  x \in \Phi_x,
\end{align}
where $\Phi_x$ is the support of $P_X$. \label{theorem:Lagrange}
\end{theorem}
It has been shown in \cite{shamaiCapacityAchieving} that the capacity-achieving input distribution is unique and discrete with a finite number of mass points for finite peak power and average power constraints. Let $P_{\bar{X}}$ denote the capacity-achieving input distribution of the DTPC under peak power and average power constraints $\pmax$ and $\pavg$, respectively. Therefore, as described in \cite{capacityAchievingDistribution} and w.l.o.g. the input distribution $P_{\bar{X}}$ is given by
\begin{align}
    P_{\bar{X}}(x) & = p_1 \delta(x-\bar{x}_1)+ p_2 \delta(x-\bar{x}_2)+\cdots + p_m\nonumber \\
    & \quad \cdot \delta(x-\bar{x}_m), \quad x \in \setx(\pmax,\pavg),
\end{align}
where $\delta(\cdot)$ denotes the Dirac impulse function, $\bar{\Phi}_x=\{\bar{x}_1,\bar{x}_2,\ldots,\bar{x}_m\}, \quad 0\leq \bar{x}_1< \bar{x}_2<\ldots<\bar{x}_m\leq \pmax $ is a finite input constellation and $\Phi_p=\{p_1,p_2,\ldots,p_m\}$ is the set of the corresponding  probability masses. Let $P_{\bar{Y}}$ denote the output distribution of the DTPC $W$ corresponding to the distribution $P_{\bar{X}}$. For all $y \in \sety$, we have
\begin{align}
    P_{\bar{Y}}(y) & = \int_{x \in \setx(\pmax,\pavg)} P_{\bar{X}}(x) W(y|x) dx \nonumber\\
    & = \int_{x \in \setx(\pmax,\pavg)} \big( \sum_{j=1}^m p_j \delta(x-\bar{x}_j) \big) W(y|x) dx \nonumber \\
    &= \sum_{j=1}^m p_j \int_{x \in \setx(\pmax,\pavg)} W(y|x) \delta(x-\bar{x}_j) dx \nonumber\\
    & = \sum_{j=1}^m p_j W(y|\bar{x}_j). \label{eq:pY}
\end{align}
Let $\bar{X}$ and $\bar{Y}$ be two RVs with probability distribution functions $P_{\bar{X}}$ and $P_{\bar{Y}}$ on $\setx(\pmax,\pavg)$ and $\sety$, respectively. In view of \eqref{eq:expectationYi}, the capacity $ C(W,\pmax,\pavg)$ of the DTPC $W$ can be rewritten as follows:
\begin{align*}
     &C(W,\pmax,\pavg)\\
      & \quad =\max_{\substack{P_X \in \mathcal{P}(\setx) \\ X \in \setx(P_{\text{max}},P_{\text{avg}})}} I(X;Y) \\
     &\quad = I(\bar{X},\bar{Y}) \\
     &\quad = \int_{x \in \setx(\pmax,\pavg)} P_{\bar{X}}(x) \sum_{y=0}^{\infty} W(y|x)(y|x_i) \log \frac{W(y|x)}{P_{\bar{Y}}(y)} dx \\
      & \quad = \int_{x \in \setx(\pmax,\pavg)} P_{\bar{X}}(x)\infdiv{W(\cdot|x)}{P_{\bar{Y}}} dx \\
     &= \quad \sum_{j=1}^m p_j \infdiv{W(\cdot|\bar{x}_j)}{P_{\bar{Y}}}.
\end{align*}
It holds for each blocklength $n$ that
\begin{align*}
    \frac{1}{n} \log \frac{W^n(Y^n|X^n)}{P_{\bar{Y}^n}(Y^n)} &= \frac{1}{n} \sum_{t=1}^n \log \frac{W(Y_t|X_t)}{P_{\bar{Y}}(Y_t)}.
\end{align*}
For fixed $x^n=(x_1,x_2,\ldots,x_n) \in \setx^n(\pmax,\pavg)$, we define
\begin{align*}
    I(\bar{Y}_t,x_t) &= \log \frac{W(Y_t|X_t)}{P_{\bar{Y}}(Y_t)}.
\end{align*}
 In addition, we have for all $t=1,2,\ldots, n$
\begin{align*}
    \mathbb{E}[I(\bar{Y}_t,x_t)] & = \infdiv{W(\cdot|x_t)}{P_{\bar{Y}}}.
\end{align*}
In order to show \eqref{eq:StrongConverseProperty}, we first prove the following inequality:
\begin{align}
    \mathbb{E}[\frac{1}{n} \sum_{t=1}^n I(\bar{Y}_i|x_i)] \leq C(W,\pmax,\pavg). \label{eq:StronConverse1}
\end{align}
It follows from Theorem \ref{theorem:Lagrange} that for $x \in [0,\pmax]$
\begin{align*}
     &L(\mu,x,P_{\bar{X}})\\
     &= I(\bar{X},\bar{Y}) + \mu(x-\pavg)-\infdiv{W(\cdot|x)}{P_{\bar{Y}}} \\
     &= C(W,\pmax,\pavg) + \mu(x-\pavg)-\infdiv{W(\cdot|x)}{P_{\bar{Y}}} \\
     & \geq 0. 
\end{align*}
As the sequence $x^n$ satisfies the peak power and average power constraints, we have
\begin{align}
  \mathbb{E}[I(\bar{Y},x)]= \infdiv{W(\cdot|x)}{P_{\bar{Y}}} & \leq  C(W,\pmax,\pavg). 
  \label{eq:expectationBound}
\end{align}
It yields that
\begin{align*}
\mathbb{E}[\frac{1}{n} \sum_{t=1}^n I(\bar{Y}_t,x_t)] & \leq C(W,\pmax,\pavg).
\end{align*}
This completes the proof of inequality \eqref{eq:StronConverse1}.\\
Next, we want to show that the variance of $I(\bar{Y}_t,x_t),\quad t=1,\ldots,n$ w.r.t. the conditional probability mass function $W(\cdot|x_t)$ for fixed $x_t \in \setx(\pmax,\pavg)$. For this purpose, we prove that $\mathbb{E}\big[(I(\bar{Y}_t,x_t))^2\big],\quad t=1,\ldots,n$ is finite. For notational simplicity, we will drop the index $t$. 
Let $\sety_1$ and $\sety_2$ be defined as
\begin{equation}
\begin{aligned}
\sety_1 & = \left\{ y \in \sety \colon 
\frac{P_{\bar{Y}}(y)}{W(y|x)} > 1 \right\}, \\
\sety_2 & = \left\{ y \in \sety \colon 
\frac{P_{\bar{Y}}(y)}{W(y|x)} \le 1 \right\}.
\end{aligned}
\end{equation}

We have
\begin{align}
\mathbb{E}\!\left[(I(\bar{Y},x))^2\right]
&= \sum_{y=0}^{\infty} W(y|x)\,
   \log^{2}\!\left( \frac{W(y|x)}{P_{\bar{Y}}(y)} \right) \nonumber\\
&= \sum_{y \in \mathcal{Y}_1} W(y|x)\,
   \log^{2}\!\left( \frac{W(y|x)}{P_{\bar{Y}}(y)} \right) \nonumber\\
&\quad + \sum_{y \in \mathcal{Y}_2} W(y|x)\,
   \log^{2}\!\left( \frac{W(y|x)}{P_{\bar{Y}}(y)} \right)
   \label{eq:sum}
\end{align}

We first establish an upper bound on the first term of the sum in \eqref{eq:sum}.
\begin{align}
& \sum_{y \in \sety_1} W(y|x) \log^2 \big( \frac{W(y|x)}{P_{\bar{Y}}(y)}\big) \nonumber\\
 & = \sum_{y\in \sety_1} W(y|x) \log^2 \big( \frac{P_{\bar{Y}}(y)}{W(y|x)}\big) \nonumber \\
 & \overset{(a)}{\leq}  \sum_{y\in \sety_1} W(y|x) \log \big( \frac{P_{\bar{Y}}(y)}{W(y|x)}\big) \cdot \log(e) \cdot \big( \frac{P_{\bar{Y}}(y)}{W(y|x)}\big) \nonumber \\
 &=  \log(e) \sum_{y\in \sety_1}  P_{\bar{Y}}(y) \log \big( \frac{P_{\bar{Y}}(y)}{W(y|x)}\big) \nonumber \\
 & \overset{(b)}{=} \log(e) \bigg(  \infdiv{P_{\bar{Y}}(\cdot)}{W(\cdot|x)} \nonumber \\
 & \quad - \sum_{y\in \sety2}  P_{\bar{Y}}(y) \log \big( \frac{P_{\bar{Y}}(y)}{W(y|x)}\big) \bigg) \nonumber \\
 &\overset{(c)}{ \leq} \log(e) \bigg(  \infdiv{P_{\bar{Y}}(\cdot)}{W(\cdot|x)} \nonumber \\
 & \quad - \sum_{y\in \sety2}  P_{\bar{Y}}(y)  \big( 1- \frac{W(y|x)}{P_{\bar{Y}}(y)}\big) \bigg) \nonumber \\
 & =  \log(e) \bigg(  \infdiv{P_{\bar{Y}}(\cdot)}{W(\cdot|x)} + \sum_{y\in \sety2} W(y|x)-P_{\bar{Y}}(y) \bigg) \nonumber \\
 & \overset{(d)}{=} \log(e) \bigg(  \infdiv{P_{\bar{Y}}(\cdot)}{W(\cdot|x)} + \frac{1}{2} d\left(P_{\bar{Y}}(\cdot),W(\cdot|x)\right) \bigg) \nonumber \\
 & \overset{(e)}{\leq}  \log(e) \bigg(  \infdiv{P_{\bar{Y}}(\cdot)}{W(\cdot|x)} +1 \bigg)\nonumber \\
 & \overset{(f)}{=}  \log(e) \bigg( \infdiv{\sum_{j=1}^m p_j W(\cdot| \bar{x}_j)}{\sum_{j=1}^m p_j W(\cdot|x)} +1 \bigg) \nonumber \\
 & \overset{(g)}{\leq} \log(e)  \bigg(\sum_{j=1}^m p_j  \infdiv{W(\cdot|\bar{x}_j}{W(\cdot|x)} + 1 \bigg) \nonumber \\
 & \overset{(h)}{\leq} \log(e) \bigg( \sum_{j=1}^m p_j \big( (\lambda+ \bar{x}_j) \log \frac{\lambda+ \bar{x}_j}{\lambda+x}  \nonumber\\
 & \quad (\bar{x}_j-x) \big) + 1 \bigg) \nonumber \\
 & \overset{(i)}{\leq} \log(e) \bigg( \sum_{j=1}^m p_j ( \lambda+ \pmax) \log \frac{\lambda+\pmax}{\lambda} + \pmax +1 \bigg) \nonumber \\
 & \overset{(j)}{\leq} \log(e) \bigg( \log(e)\frac{(\lambda+\pmax)^2}{\lambda}+\pmax +1 \bigg), \label{eq:firstterm}
\end{align}
where $(a)$ follows because $\log \big( \frac{P_{\bar{Y}}(y)}{W(y|x)}\big) >0$ and $\log(x) \leq (x-1)\log(e)$, $(b)$ follows from the definition of the Kullback-Leibler divergence, $(c)$ follows because $\log(x)\geq \log(e)(1-\frac{1}{x})$, $(d)$ follows from the definition of the variational distance $d\left(P_{\bar{Y}}(\cdot),W(\cdot|x)\right)$, 
$(e)$ follows because $0\leq d\left(P_{\bar{Y}}(\cdot),W(\cdot|x)\right) \leq 2$,
$(f)$ follows from \eqref{eq:pY}, $(g)$ follows from the convexity of the Kullback-Leibler divergence,
$(h)$ because $W(\cdot|\bar{x}_j), \quad j=1,\ldots,m$ and $W(\cdot|x)$ are Poisson distributed with mean $\lambda+\bar{x}_j$ and $\lambda+x$, respectively, 
$(i)$ follows because $0\leq x, \ \bar{x}_j\leq \pmax$ and $(j)$ follows because $\log(x) \leq (x-1)\log(e)$. \\
Now, we compute an upper bound on the second term of the sum in \eqref{eq:sum}. We have
\begin{align}
& \sum_{y \in \sety_2} W(y|x) \log^2 \big( \frac{W(y|x)}{P_{\bar{Y}}(y)}\big) \nonumber\\
&=  \sum_{y \in \sety_2} W(y|x) \big( \log(W(y|x)) - \log(P_{\bar{Y}}(y)) \big)^2 \nonumber \\
& \overset{(a)}{=}  \sum_{y \in \sety_2} W(y|x) \bigg( \log(W(y|x)) \nonumber \\
& \quad - \log\big(\sum_{j=1}^m p_j W(\cdot| \bar{x}_j)\big) \bigg)^2 \nonumber \\
&\overset{(b)}{\leq}  \sum_{y \in \sety_2} W(y|x) \bigg( \log(W(y|x)) \nonumber \\
& \quad - \sum_{j=1}^m p_j \log\big(W(\cdot| \bar{x}_j)\big)  \bigg)^2 \nonumber \\
& \overset{(c)}{=}  \sum_{y \in \sety_2} W(y|x) \bigg( \big(\log(x+\lambda) -  \sum_{j=1}^m p_j \log(x_j+\lambda) \big)y \nonumber \\
& \quad + \big(  \sum_{j=1}^m p_jx_j - x \big) \bigg)^2 \nonumber \\
&\overset{(d)}{\leq} \sum_{y \in \sety_2} W(y|x) \bigg( \alpha y + \beta \bigg)^2 \nonumber \\
&= \alpha^2\sum_{y \in \sety_2} W(y|x) y^2 + 2 \alpha  \beta \sum_{y \in \sety_2}  W(y|x) y \nonumber \\
& \quad + \beta^2 \sum_{y \in \sety_2} W(y|x) \nonumber \\
& \leq \alpha^2\mathbb{E}[Y^2|X=x]+  2 \alpha\nonumber  \beta \mathbb{E}[Y|X=x] + \beta^2(\lambda,\pmax) \nonumber \\
& \overset{(e)}{\leq}  \alpha^2 (\lambda+\pmax)^2 + 2 \alpha  \beta (\lambda+\pmax) +  \beta^2,\label{eq:secondterm}
\end{align}
where $d(\cdot,\cdot)$ denotes the total variational distance, $(a)$ follows \eqref{eq:pY}, $(b)$ follows from the convexity of $-\log(x)$, $(c)$ follows because $W(\cdot|\bar{x}_j), \quad j=1,\ldots,m$ and $W(\cdot|x)$ are Poisson distributed with mean $\lambda+\bar{x}_j$ and $\lambda+x$, respectively, $(d)$ follows from \eqref{eq:alpha} and \eqref{eq:beta} and $(e)$ follows because given the input $X=x$, the output $Y$ is Poisson distributed with mean $\lambda+x$ and $0<x\leq \pmax$. \\
Let $\alpha$ and $\beta$ be defined as follows:
\begin{align*}
    \alpha &= \log(1+\frac{\pmax}{\lambda} ) \\
    \beta & = \pmax.
\end{align*}
Then, it can be easily shown that 
\begin{align}
\log(x+\lambda) -  \sum_{j=1}^m p_j \log(x_j+\lambda) & \leq \alpha,
\label{eq:alpha}
\end{align}
and 
\begin{align}
    \sum_{j=1}^m p_jx_j - x & \leq \beta. \label{eq:beta}
\end{align}
It follows from \eqref{eq:firstterm} and \eqref{eq:secondterm} that
 \begin{align*}
    &{\text{Var}}[\frac{1}{n} \sum_{t=1}^n I(\bar{Y}_t,x_t)] \nonumber\\
    & \leq \log(e) \bigg( \log(e)\frac{(\lambda+\pmax)^2}{\lambda}+\pmax +1 \bigg) \nonumber \\
    & \quad +  \alpha^2 (\lambda+\pmax)^2 + 2 \alpha  \beta (\lambda+\pmax) +  \beta^2 .\\
\end{align*}
Let $\gamma(\lambda,\pmax)$ be defined as 
\begin{align*}
    \gamma(\lambda,\pmax) & = \log(e) \bigg( \log(e)\frac{(\lambda+\pmax)^2}{\lambda}+\pmax +1 \bigg) \\
    & \quad +  \alpha^2 (\lambda+\pmax)^2 + 2 \alpha  \beta (\lambda+\pmax) +  \beta^2.
\end{align*}
Therefore, Chebyshev's inequality implies
\begin{align}
    &\Pr\big\{ \frac{1}{n} \sum_{t=1}^n I(\bar{Y}_t,x_t) \geq C(W,\pmax,\pavg)+ \nu | X^n=x^n \big \} \nonumber\\
    & \quad \leq  \frac{\gamma(\lambda,\pmax)}{n}, \label{eq:tshebychev}
\end{align}
where $\nu>0$ is an arbitrary constant. As $\mathbf{X}=\{X^n\}_{n=1}^\infty \in \mathcal{U}(\pmax,\pavg)$ is assumed, \eqref{eq:tshebychev} holds for all realizations $x^n$ of $X^n$. Thus, we have
\begin{align*}
    &\Pr\big\{ \frac{1}{n} \sum_{t=1}^n I(\bar{Y}_t,X_t) \geq C(W,\pmax,\pavg)+ \nu  \big \} \\
    & \quad \leq  \frac{\gamma(\lambda,\pmax)}{n}.
\end{align*}
That means
\begin{align*}
    &\Pr\big\{ \frac{1}{n} \log \frac{W^n(Y^n|X^n)}{P_{\bar{Y}^n}(Y^n)} \geq C(W,\pmax,\pavg)+ \nu  \big \} \\ & \quad \leq  \frac{\gamma(\lambda,\pmax)}{n}. 
\end{align*}
We have
\begin{align*}
    &\lim_{n \to \infty} \Pr\big\{ \frac{1}{n} \log \frac{W^n(Y^n|X^n)}{P_{\bar{Y}^n}(Y^n)} \geq C(W,\pmax,\pavg)+ \nu  \big \} \\
    & \quad =0.
\end{align*}
Since $\nu>0$ is chosen arbitrarily, we have
\begin{align}
    \plimsup_{n \to \infty} \frac{1}{n} \log \frac{W^n(Y^n|X^n)}{P_{\bar{Y}^n}(Y^n)} & \leq  C(W,\pmax,\pavg). \label{eq:lastIneq}
\end{align}
It follows from \eqref{eq:lastIneq} and \cite[Lemma 3.2.1]{HanBook} that
\begin{align*}
    \plimsup_{n \to \infty} \frac{1}{n} \log \frac{W^n(Y^n|X^n)}{P_{{Y}^n}(Y^n)} & \leq  C(W,\pmax,\pavg). \label{eq:lastIneq}
\end{align*}
This completes the proof of \eqref{eq:StrongConverseProperty}. Thus, \eqref{eq:Strong} holds, and the strong converse property is proved for the DTPC $W$ under peak and average power constraints $\pmax$ and $\pavg$, respectively. That means that the RI capacity of the DTPC $C_{ID}(W,\pmax,\pavg)$ coincides with its transmission capacity $C(W,\pmax,\pavg)$.

\section{Proof of Theorem \ref{theorem:DichotomyTheorem}}\label{sec:Secureproof}
A transmission code \( C' \) is used to encode the colors, while a wiretap code \( C'' \) is used to encode the color.

$C' = \{(u'_j, \mathcal{D}'_j), j \in \{1, \dots, M'\}$, $u'_j \in \mathcal{X}', \mathcal{D}'_j \subset \mathcal{Y}' \} $, with rate $R'_c = C(W,P_{max}, P_{avg}) - \epsilon$, $ 0 \leq \epsilon \leq C(W,P_{max}, P_{avg})$ and probability of vanishing error $\lambda'_n$, $\lim_{n \to \infty} \lambda'_n = 0$.

\begin{equation}
\sum_{l=1}^{n} u'_{j,l} \le n \cdot P_{avg}+P_{max}.
\end{equation}

\begin{figure*}[htb]
    \centering
    \resizebox{0.6\linewidth}{!}{\tikzset{every picture/.style={line width=0.75pt}} 

\begin{tikzpicture}[x=0.75pt,y=0.75pt,yscale=-1,xscale=1]

\draw   (120.5,146.25) -- (171,146.25) -- (171,346.17) -- (120.5,346.17) -- cycle ;
\draw    (120,195.6) -- (171,195.75) ;
\draw    (120.25,246.14) -- (171.25,246.28) ;
\draw    (120.25,295.85) -- (171.25,296) ;
\draw   (507.83,144.59) -- (558.33,144.59) -- (558.33,344.5) -- (507.83,344.5) -- cycle ;
\draw    (507.33,193.94) -- (558.33,194.09) ;
\draw    (507.58,244.47) -- (558.58,244.62) ;
\draw    (507.58,294.18) -- (558.58,294.33) ;
\draw   (256.67,45.33) -- (386.67,45.33) -- (386.67,92.18) -- (256.67,92.18) -- cycle ;
\draw    (336.47,45.4) -- (336.87,92.11) ;
\draw   (257.07,96.27) .. controls (257.07,100.94) and (259.4,103.27) .. (264.07,103.27) -- (286.27,103.27) .. controls (292.94,103.27) and (296.27,105.6) .. (296.27,110.27) .. controls (296.27,105.6) and (299.6,103.27) .. (306.27,103.27)(303.27,103.27) -- (328.47,103.27) .. controls (333.14,103.27) and (335.47,100.94) .. (335.47,96.27) ;
\draw   (336.87,96.67) .. controls (336.83,101.34) and (339.14,103.69) .. (343.81,103.73) -- (351.61,103.79) .. controls (358.28,103.84) and (361.59,106.2) .. (361.55,110.87) .. controls (361.59,106.2) and (364.94,103.9) .. (371.61,103.95)(368.61,103.93) -- (379.41,104.01) .. controls (384.08,104.05) and (386.43,101.74) .. (386.47,97.07) ;
\draw   (386.17,41.42) .. controls (386.18,36.75) and (383.85,34.42) .. (379.18,34.41) -- (332.66,34.35) .. controls (325.99,34.34) and (322.66,32.01) .. (322.67,27.34) .. controls (322.66,32.01) and (319.33,34.34) .. (312.66,34.33)(315.66,34.33) -- (265.98,34.27) .. controls (261.31,34.26) and (258.98,36.59) .. (258.97,41.26) ;
\draw   (569,344) .. controls (573.67,344) and (576,341.67) .. (576,337) -- (576,254) .. controls (576,247.33) and (578.33,244) .. (583,244) .. controls (578.33,244) and (576,240.67) .. (576,234)(576,237) -- (576,151) .. controls (576,146.33) and (573.67,144) .. (569,144) ;
\draw   (103.67,146.67) .. controls (99,146.64) and (96.65,148.95) .. (96.62,153.62) -- (96.07,236.29) .. controls (96.03,242.96) and (93.68,246.27) .. (89.01,246.24) .. controls (93.68,246.27) and (95.99,249.62) .. (95.94,256.29)(95.96,253.29) -- (95.39,338.95) .. controls (95.36,343.62) and (97.67,345.97) .. (102.34,346) ;
\draw    (170.67,270) -- (252.87,80.83) ;
\draw [shift={(253.67,79)}, rotate = 113.49] [color={rgb, 255:red, 0; green, 0; blue, 0 }  ][line width=0.75]    (10.93,-3.29) .. controls (6.95,-1.4) and (3.31,-0.3) .. (0,0) .. controls (3.31,0.3) and (6.95,1.4) .. (10.93,3.29)   ;
\draw    (506.67,268) -- (390.7,76.71) ;
\draw [shift={(389.67,75)}, rotate = 58.78] [color={rgb, 255:red, 0; green, 0; blue, 0 }  ][line width=0.75]    (10.93,-3.29) .. controls (6.95,-1.4) and (3.31,-0.3) .. (0,0) .. controls (3.31,0.3) and (6.95,1.4) .. (10.93,3.29)   ;

\draw (137.5,214.57) node [anchor=north west][inner sep=0.75pt]    {$u'_{2}$};
\draw (136,163.07) node [anchor=north west][inner sep=0.75pt]    {$u'_{1}$};
\draw (134.5,311.57) node [anchor=north west][inner sep=0.75pt]    {$u'_{M'}$};
\draw (140,255) node [anchor=north west][inner sep=0.75pt]    {$\vdots $};
\draw (524.83,212.9) node [anchor=north west][inner sep=0.75pt]    {$u''_{2}$};
\draw (523.33,161.4) node [anchor=north west][inner sep=0.75pt]    {$u''_{1}$};
\draw (521.83,309.9) node [anchor=north west][inner sep=0.75pt]    {$u''_{M''}$};
\draw (530,255) node [anchor=north west][inner sep=0.75pt]    {$\vdots $};
\draw (288.97,117.67) node [anchor=north west][inner sep=0.75pt]    {$n$};
\draw (344,112.07) node [anchor=north west][inner sep=0.75pt]    {$\left\lceil \sqrt{n} \right\rceil$};
\draw (196.67,2.73) node [anchor=north west][inner sep=0.75pt]    {$m\in \ C=\{( Q_{i} ,\ \mathcal{D}_{i}) ,\ i\ \in \ \{1,\ ...,\ N\}\}$};
\draw (585,229.07) node [anchor=north west][inner sep=0.75pt]    {$\left| C'' \right| = \left\lceil 2^{\sqrt{n} \epsilon} \right\rceil
$};
\draw (422,360) node [anchor=north west][inner sep=0.75pt]    {$C''=\{( u''_{k} ,\ \mathcal{D} ''_{k}) ,\ k\ \in \ \{1,\ ...,\ M''\}\}$};
\draw (20.67,359.4) node [anchor=north west][inner sep=0.75pt]    {$C'=\{( u'_{j} ,\ \mathcal{D} '_{j}) ,\ j\ \in \ \{1,\ ...,\ M'\}\}$};
\draw (-10,237.07) node [anchor=north west][inner sep=0.75pt]    {$\left| C' \right| = \left\lceil 2^{n(C - \epsilon)} \right\rceil
$};
\draw (90,109) node [anchor=north west][inner sep=0.75pt]   [align=left] {Transmission code};
\draw (495,109) node [anchor=north west][inner sep=0.75pt]   [align=left] {Wiretap code};

\end{tikzpicture}} 
    \caption{Code construction for the wiretap channel.}
    \label{fig:CodeConstruction}
\end{figure*}
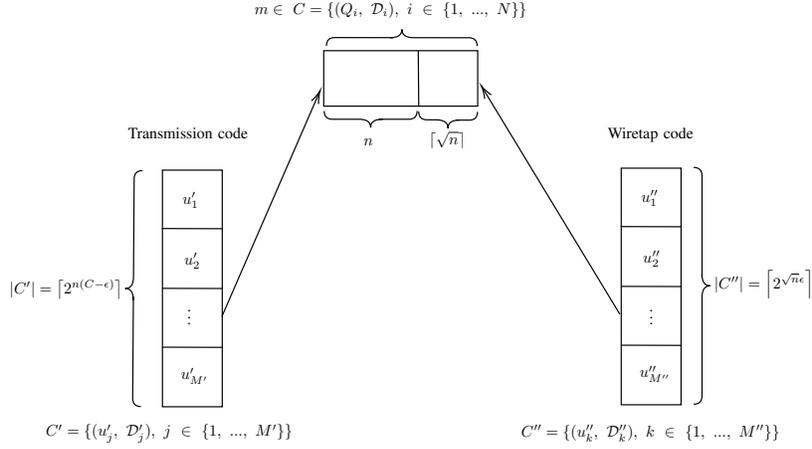

$C'' = \{(u''_k, \mathcal{D}''_k), k \in \{1, \dots, M''\}$, $u''_k \in \mathcal{X}'', \mathcal{D}''_k \subset \mathcal{Y}'' \} $, is a wiretap code with a rate $R''_c = \epsilon, \, 0 \le \epsilon \le C(W,P_{max}, P_{avg})$ and a probability of vanishing error \( \lambda''_n \), \(\lim_{n \to \infty} \lambda''_n = 0.\)

Every codeword \( u''_k = (u''_{k,1}, u''_{k,2}, \ldots, u''_{k,n}) \) fulfills the following power constraint:

\begin{equation}
\sum_{l=1}^{n} u''_{k,l} \leq \lceil \sqrt{n} \rceil \cdot P_{avg}+P_{max}.
\end{equation}

The codes \( C' \) and \( C'' \) are then concatenated to construct an identification code \((m, N, \lambda_1, \lambda_2)\) for \((W, V, P_{max}, P_{avg})\) given by
$C = \{(Q(\cdot|i), \mathcal{D}_i), Q(\cdot|i) \in \mathcal{P}(\mathcal{X}^n(P_{max}, P_{avg})), \mathcal{D}_i \subset \mathcal{Y}^n, i \in \{1, \dots, N\}\}$.

with \(\lambda_1, \lambda_2 \leq \lambda < \frac{1}{2}\) and \(m = n + \lceil \sqrt{n} \rceil\). The goal is to show that:

\begin{multline}
\label{eq:CSID}
C_{SID}(W, V, P_{max}, P_{avg}) = C(W, P_{max}, P_{avg}) \\
\text{if } C_{S}(W, V, P_{max}, P_{avg}) > 0.
\end{multline}

 It must be shown that equation \ref{eq:Eve} holds. For this purpose, it is sufficient to prove that the wiretapper cannot identify the color, i.e., the second fundamental component of the transmitted codeword.

For a discrete wiretap channel, \( Q_i V^p(\varepsilon) \) is defined as follows:
\begin{equation}
Q_i V^p(\varepsilon) \triangleq \sum_{x^p \in X^p} Q(x^p|i) V^p(\varepsilon|x^p).
\end{equation}

For any region \( \varepsilon \in \mathbb{\mathcal Z}^p \) and any \( i \neq j \), the variational distance between the two output measures \( Q_i V^p(\varepsilon) \) and \( Q_j V^p(\varepsilon) \) is upper bound. From the transmission wiretap code, a transmission code can be constructed that satisfies the following inequality:

\begin{equation}
    |Q_i V^p(\varepsilon) - Q_j V^p(\varepsilon)| \leq \epsilon.
\end{equation}

\begin{definition}
For a discrete Poisson wiretap channel, \( Q_i V^p(\varepsilon) \) is defined as follows:

\begin{equation}
    Q_i V^p(\varepsilon) \equiv \int_{x^p \in \mathcal X^p} Q(x^p|i) V^p(\varepsilon|x^p) \, d^p x^p.
\end{equation}
The original Poisson wiretap channel \( (W, V, P_{\text{max}}, P_{\text{avg}}) \) is approximated by a discrete wiretap channel \( (W, V, P_{\text{max}}, P_{\text{avg}}, \mathcal L_x, z_0) \). The channels \( W(P_{\text{max}}, P_{\text{avg}}) \) and \( V(P_{\text{max}}, P_{\text{avg}}) \) are also approximated by the discrete channels \( W(P_{\text{max}}, P_{\text{avg}}, \mathcal L_x, z_0) \) and \( V(P_{\text{max}}, P_{\text{avg}}, \mathcal L_x, z_0) \), respectively. If \( Q_i V^p(\cdot) \) is an output measure on \( V(P_{\text{max}}, P_{\text{avg}}) \), its approximation on \( V(P_{\text{max}}, P_{\text{avg}}, \mathcal L_x, z_0) \) is denoted by \( \hat{Q}_i V^p(\cdot) \). It is recalled that after quantization, the resulting output measure closely approximates the original one.
\begin{equation}
    \label{eq:one}
    \left| Q_i V^p(\cdot) - \hat{Q}_i V^p(\cdot) \right| \leq \delta', \quad \delta' > 0.
\end{equation}

For any region \( \varepsilon \in \mathcal{Z}^p \):
\begin{equation}
    \label{eq:two}
    \left| \hat{Q}_i V^p(\varepsilon) - \hat{Q}_j V^p(\varepsilon) \right| \leq \epsilon.
\end{equation}

\end{definition}

It follows from equations \ref{eq:one} and \ref{eq:two} that:
\begin{equation}
    \label{eq:last}
    \left| Q_i V^m(\varepsilon) - Q_j V^m(\varepsilon) \right| \leq \epsilon + 2\delta'.
\end{equation}
\( \delta' \) is chosen small enough such that \( \epsilon + 2\delta' = \epsilon_1 \) is very small. It is clear that equation \ref{eq:last} implies the last secrecy requirement of an identification code for the DTPWC \( (W, V, P_{\text{max}}, P_{\text{avg}}) \).

\section{Conclusions} \label{sec:conclusion}

In this work, we present an information-theoretic analysis of event-triggered molecular communication. Previous studies known to us have considered this model under the assumption of \emph{deterministic encoding}. In contrast, our analysis assumes that local randomness is available at the sender. While the deterministic case yields $n^{nR_1}$ possible messages, the randomized setting enables $2^{2^{nR_2}}$ possible messages.

Moreover, we derive an exact expression for the RI capacity, whereas in the deterministic identification (DI) case, only upper and lower bounds have been established thus far. Naturally, one may question the realism of assuming available randomness at the sender. A plausible scenario includes a sender located outside the body (e.g., in a wearable device such as a bracelet). Alternatively, randomness may be obtained through additional resources such as feedback.
We further examine a setting where the channel behavior depends on a random state and derive the corresponding capacity.

Additionally, we consider the DTPWC for secure identification. In this context, we perform a security analysis for the degraded wiretap channel. A potential adversary may be modeled as a malicious nanomachine or cell positioned downstream of the legitimate nanomachine, attempting to eavesdrop on the communication. This highlights the advantage of randomized identification (RI): as long as secure communication is possible, we can achieve the same identification capacity as in the case without any secrecy constraint.

We showed that the identification and secure identification capacities exhibit the same fundamental behavior as the transmission capacity of the DTPC. Deriving numerical results for these capacity expressions is highly challenging due to the underlying optimization problems. However, several works focus specifically on this optimization challenge. Since we established that the capacity formulas coincide in our setting, the numerical results obtained in those studies directly apply to our case. Therefore, we refer to previous work, such as \cite{4729780, li2025design}, which provides numerical evaluations for special classes of discrete-time Poisson channels.

Future work should investigate deterministic identification augmented with resources that enable randomized encoding. Another fundamental direction is the study of finite blocklength encoding strategies for this model. 

Moreover, in this work, the ISI effect was intentionally neglected to maintain analytical tractability and to highlight the core principles of secure identification. Future extensions will generalize the framework to Poisson molecular channels with ISI and molecular noise, following the approaches in \cite{salariseddigh2023deterministic} and \cite{zhao2025identification}. Such generalization will provide a more comprehensive and realistic description of secure molecular identification systems.

Our work constitutes a first step in this line of research.


\section*{Acknowledgments}
The authors acknowledge the financial support by the Federal Ministry of Research, Technology and Space (BMFTR)
in the programme of “Souverän. Digital. Vernetzt.”. Joint
project 6G-life, project identification number: 16KISK002
and 16KISK263. H. Boche and W. Labidi were further sup-
ported in part by the BMFTR within the national initiative
on Molecurlare Communication IoBNT under Grant
16KIS1988. C. Deppe was further supported in part by the
BMFTR within the national initiative on Post Shannon Com-
munication (NewCom) under Grant 16KIS1005. C. Deppe, V.
Gholamian, Y. Zhao were also supported by the DFG within
the project DE1915/2-1.

\bibliographystyle{IEEEtran}
\bibliography{definitions,references}
\begin{IEEEbiography}[{\includegraphics[width=1in,height=1.25in,clip]{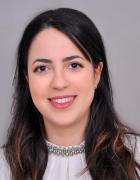}}]
    {Wafa Labidi} (Graduate Student Member, IEEE) was born in Tunis, Tunisia on April 1, 1993. She received the B.Sc. and M.Sc. degrees in electrical engineering from the Technical University of Munich (TUM), Germany, in 2016 and 2019, respectively. She is currently working toward the Ph.D. degree at the Chair of Theoretical Information Technology, TUM. Her research mainly focuses on message identification, common randomness generation and molecular communication.
\end{IEEEbiography}
\begin{IEEEbiography}
[{\includegraphics[width=1in,height=1.25in,clip]{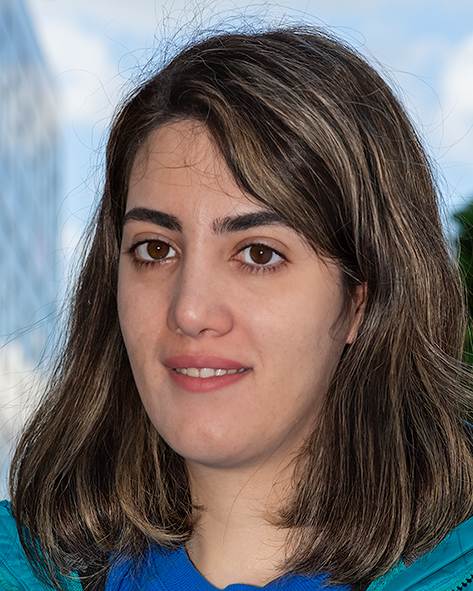}}]
    {Vida Gholamian} was born in Ghazvin, Iran. After completing high school, she developed a strong interest in programming and mathematics, which led her to pursue a Bachelor’s degree in Information Technology Engineering at the Institute for Advanced Studies in Basic Sciences (IASBS) in Zanjan. She completed her Bachelor’s degree in 2019 with a focus on the Internet of Things (IoT). In the same year, she began her Master’s studies in Computer Science at IASBS. Following her graduation, she worked as a Research Assistant in the Wireless Networks Laboratory at IASBS. She later expanded her research interests to include post-Shannon identification and molecular communication. She subsequently joined the Department of Information Theory and Communications Technology at the Institute for Communications Technology (IfN), where she is currently pursuing her Ph.D. as a research assistant.
\end{IEEEbiography}
\begin{IEEEbiography}
[{\includegraphics[width=1in,height=1.25in,clip]{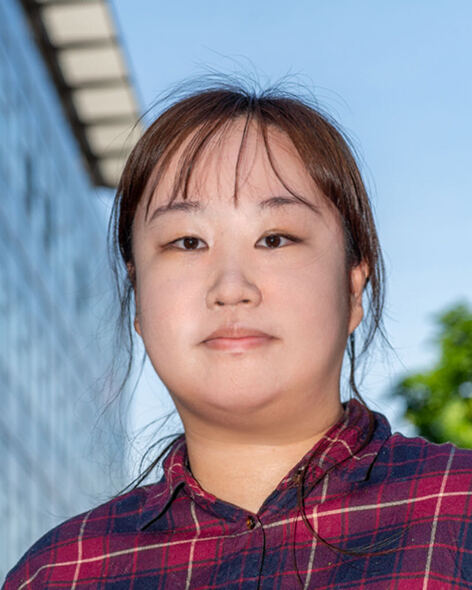}}]
    {Yaning Zhao} was born in Shandong, China. She received her Bachelor’s degree in Telecommunications Engineering from Xidian University (Xi’an University of Electronic Science and Technology) in 2020. She then pursued a Master’s degree in Communications Engineering at Technische Universität München. During her studies, she worked as a part-time research assistant, conducting research on the interpretability of neural networks. Her research interests later shifted toward multiuser information theory and post-Shannon theory, which culminated in her master’s thesis on joint identification and sensing in multiuser systems at the Chair of Theoretical Information Technology. In 2024, she relocated to Braunschweig to continue her academic career at as phd student Technische Universität Braunschweig. 
\end{IEEEbiography}
\begin{IEEEbiography}
[{\includegraphics[width=1in,height=1.25in,clip]{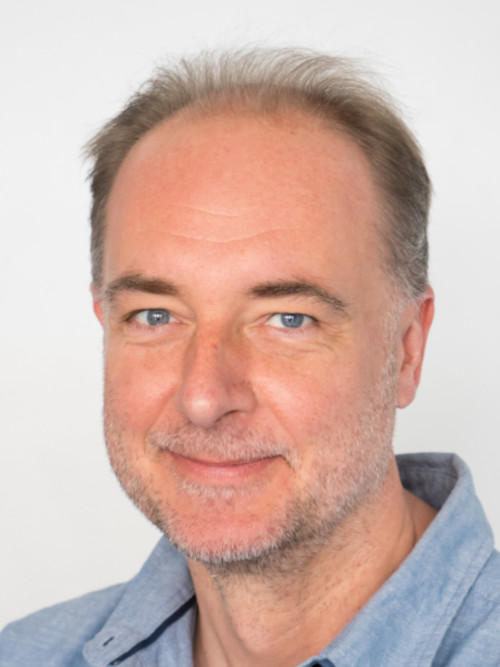}}]
    {Christian Deppe} (Senior Member, IEEE) received his Dipl.-Math. degree in mathematics from the University of Bielefeld in 1996 and his Dr.-math. degree, also from the University of Bielefeld, in 1998. He then worked there until 2010 as a research associate and assistant at the Faculty of Mathematics, Bielefeld. In 2011, he took over the management of the project "Safety and Robustness of the Quantum Repeater" from the Federal Ministry of Education and Research at the Faculty of Mathematics, Bielefeld University, for two years. In 2014, Christian Deppe was funded by a DFG project at the Chair of Theoretical Information Technology, Technical University of Munich. At the Friedrich Schiller University in Jena, Christian Deppe took up a temporary professorship at the Faculty of Mathematics and Computer Science in 2015. Until 2023, he worked for six years at the Chair of Communications Engineering at the Technical University of Munich and since January 2024 has taken on new tasks at the Institute of Communications Engineering at the TU Braunschweig. He is project leader of several projects funded by the BMFTR and the DFG.
\end{IEEEbiography}
\begin{IEEEbiography}
[{\includegraphics[width=1in,height=1.25in,clip]{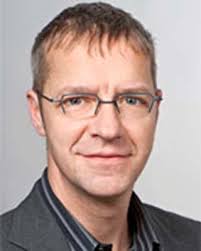}}]
    {Holger Boche} (Fellow, IEEE) received the Dipl.-Ing. degree in electrical
engineering, the degree in mathematics, and the Dr.-Ing. degree in electrical
engineering from the Technische Universität Dresden, Dresden, Germany,
in 1990, 1992, and 1994, respectively, the master’s degree from the Friedrich-
Schiller Universität Jena, Jena, Germany, in 1997, and the Dr.Rer.Nat. degree
in pure mathematics from the Technische Universität Berlin, Berlin, Germany,
in 1998.
In 1997, he joined the Fraunhofer Institute for Telecommunications,
Heinrich-Hertz-Institute (HHI), Berlin. From 2002 to 2010, he was a Full
Professor of mobile communication networks with the Institute for Commu-
nications Systems, Technische Universität Berlin. In 2003, he became the
Director of the Fraunhofer German-Sino Laboratory for Mobile Communica-
tions, Berlin. In 2004, he became the Director of the Fraunhofer Institute for
Telecommunications, HHI. He was a Visiting Professor with ETH Zürich,
Zurich, Switzerland, from 2004 to 2006 (Winter); and KTH Stockholm,
Stockholm, Sweden, in 2005 (Summer). He has been a member and an
Honorary Fellow of the TUM Institute for Advanced Study, Munich, Germany,
since 2014. Since 2018, he has been the Founding Director of the Center for
Quantum Engineering, Technische Universität Muünchen, Munich, Germany.
Since 2021, he has been jointly leading the BMBF Research Hub 6G-Life with
Frank Fitzek. He is currently a Full Professor with the Institute of Theoretical
Information Technology, Technische Universität Muünchen, which he joined
in October 2010. Among his publications is the recent book, Information
Theoretic Security and Privacy of Information Systems (Cambridge University
Press, 2017).
Prof. Boche is a member of the IEEE Signal Processing Society SPCOM
and SPTM Technical Committees. He was an Elected Member of the German
Academy of Sciences (Leopoldina) in 2008 and the Berlin Brandenburg
Academy of Sciences and Humanities in 2009. He was a recipient of
the Research Award “Technische Kommunikation” from the Alcatel SEL
Foundation in October 2003, the “Innovation Award” from the Vodafone
Foundation in June 2006, and the Gottfried Wilhelm Leibniz Prize from the
Deutsche Forschungsgemeinschaft (German Research Foundation) in 2008.
He was a co-recipient of the 2006 IEEE Signal Processing Society Best Paper
Award and a recipient of the 2007 IEEE Signal Processing Society Best Paper
Award. He was the General Chair of the Symposium on Information Theoretic
Approaches to Security and Privacy at IEEE GlobalSIP 2016.

\end{IEEEbiography}
\clearpage
\newpage

\IEEEtriggeratref{4}



\end{document}